\documentclass[prd,superscriptaddress,nofootinbib,notitlepage]{revtex4-1}
\pdfoutput=1
\usepackage[colorlinks=true,citecolor=blue,linkcolor=black]{hyperref}
\usepackage{epsfig,amsmath,amssymb,slashed}
\usepackage{graphicx}
\usepackage{color}
\usepackage{siunitx}
\usepackage{color}
\usepackage{bm}
\usepackage{subcaption}
\usepackage{tabularx}
\usepackage{cleveref}
\usepackage{float}
\usepackage{amsthm} 
\usepackage{enumitem} 

\theoremstyle{definition}
\newtheorem{theorem}{Theorem}
\newtheorem{definition}{Definition}
\newtheorem{corollary}{Corollary}

\DeclareMathOperator*{\dist}{dist}

\newcommand{\bra}[1]{\langle #1|}
\newcommand{\ket}[1]{|#1\rangle}

\newcommand{\di}{{\rm d}}
\newcommand{\ii}{i}

 \def\wT{{\widehat T}}

\def\wO{{\widehat O}}

\def\wphi{{\widehat{\phi}}}

\def\wrho{{\widehat{\rho}}}

\newcommand{\Tr}{{\rm Tr}}  
  
\newcommand{\e}{{\rm e}}

\newcommand{\be}{\begin{equation}}
\newcommand{\ee}{\end{equation}}                                                                               
\def\bea{\begin{eqnarray}}
\def\eea{\end{eqnarray}}

\begin{document}
%

\title{Universal analytic dependence of the stress-energy tensor at thermodynamic equilibrium in curved space-time} 

\author{F. Becattini}
\affiliation{Universit\`a degli studi di Firenze and INFN Sezione di Firenze,\\
Via G. Sansone 1, I-50019 Sesto Fiorentino (Florence), Italy}

\author{F. Palli}
\affiliation{Universit\`a degli studi di Firenze and INFN Sezione di Firenze,\\
Via G. Sansone 1, I-50019 Sesto Fiorentino (Florence), Italy}

\begin{abstract}
The mean value of the stress-energy tensor of a given quantum field theory at global thermodynamic equilibrium in a curved 
space-time can be expressed in terms of the derivatives of the Killing four-temperature field and the derivatives of the metric tensor.
Its asymptotic expansion about zero includes an analytic part made of integer powers of these derivatives - corresponding
to the so-called gradient expansion - as well as non-perturbative corrections. By using available exact solutions for the free real 
massless scalar field, we show that in the case of Minkowski, de Sitter, anti-de Sitter, and closed Einstein universe, the analytic 
part - obtained through the procedure of analytic distillation - has a finite number of terms and it is the same once expressed 
in a covariant form. On the other hand, non-universal terms are non-analytic in these derivatives and correspond to boundary 
conditions or to specific global properties of the space-time. We argue that the universality of the analytic part extends to 
any quantum field theory on a curved background. 
\end{abstract}

\maketitle

\section{Introduction}
\label{sec:intro}

The calculation of the stress-energy tensor at global thermodynamic equilibrium in curved space-time has been the subject 
of several investigations over some decades \cite{Allen:1986ty, Kennedy:1979ar, Ambrus:2018olh, Panerai:2015xlr}. Various 
approximation methods have been developed to study its properties in static space-times \cite{Page:1982fm, Brown:1985ri, Brown:1986jy, Frolov:1987gw} 
and the high-temperature regime has been explored \cite{Altaie:1978dx, Dowker:1988jw, Fursaev:2001yu}. 
Particularly, the response of the stress-energy tensor to the curvature was studied with perturbative methods 
\cite{Moore:2012tc,Kovtun:2018dvd}. The perturbative expansion is in fact a so-called gradient expansion, that is an asymptotic 
expansion in the derivatives of the four-temperature (including gradients of temperature and flow velocity) and of the metric. 
Central to the idea of the gradient expansion is the tacit assumption that its coefficients are universal, i.e. they do not 
depend on the particular space-time under consideration, so there is an actual thermodynamic response of the stress-energy tensor 
to the curvature, acceleration and vorticity.
However, this assumption has not been verified by studying the expansion of an otherwise found exact solution in
curved space-time. The goal of this paper is to fill this gap, namely to show that the asymptotic expansion of available
exact solutions for the real massless free scalar field in the derivatives of the metric and the four-temperature (which is 
a Killing vector field at equilibrium) are indeed universal provided that these derivatives are regarded as complex variables. 

In our analysis, we will consider a real free massless scalar field, either conformally ($\xi=1/6)$ or minimally $(\xi=0)$ coupled 
to gravity, described by the following action:
\begin{equation*}
    \widehat S_\xi=\frac{1}{2}\int \di^4x \; \sqrt{-g}\left(g^{\mu\nu}\nabla_{\mu}\widehat{\phi}
    \, \nabla_{\nu}\widehat{\phi}- \xi R\widehat{\phi}^{\ 2}\right).
\end{equation*}
For this field, we will use the exact renormalized thermal expectation value of the stress-energy tensor operator:
\begin{align*}
    \widehat{T}_{\mu\nu} =& \;\left(1-2\xi\right)\nabla_\mu \wphi \; 
     \nabla_\nu\wphi + \left(2\xi-\frac{1}{2}\right)g_{\mu\nu}g^{\rho\sigma}\nabla_\rho \wphi \; 
     \nabla_\sigma\wphi -2\xi \wphi \; \nabla_{\mu}\nabla_{\nu}\wphi \; +2\xi g_{\mu\nu}\wphi \; \Box \; \wphi - \xi G_{\mu\nu} \wphi^{\ 2}
\end{align*}
where $G_{\mu\nu}$ is the Einstein tensor. We will check that the coefficients of the asymptotic expansion of the expectation value
of the stress-energy tensor proportional to integer powers of curvature, acceleration and vorticity filtered by the so-called
analytic distillation are independent of the particular space-time in four cases: Minkowski space-time, de Sitter, anti-de Sitter and 
closed Einstein universe. This result confirms the aforementioned tacit assumption on the gradient expansion and suggests that this 
universality extends to {\em any} quantum field theory. On the other hand, we will see that the coefficients of non-integer powers of
derivatives appearing in the asymptotic expansion, which are non-analytic, are not universal, i.e. they do depend on the particular 
space-time. 

The paper is organized as follows: in Section \ref{sec:gte} the basic concepts of thermodynamic equilibrium in curved
space-time are summarized; in Section \ref{sec:conj} the universal analytic conjecture is presented and the method of
analytic distillation is discussed in some detail; in Section \ref{sec:mink} we report the known exact solution in 
Minkowski space-time; in Sections \ref{sec:ads},\ref{sec:ds} we report the exact solution known in anti-de Sitter and
de Sitter space-time and we apply analytic distillation to extract its analytic part; in Section \ref{sec:ceu} we do
the same for the closed Einstein universe; in Section \ref{sec:uni} we recast the obtained analytic distillates in a
covariant form to be properly compared with each other and with the Minkowski space-time solution; in Section \ref{sec:unruh}
we discuss the relation between our analytic distillates and the Unruh temperature in various space-times; finally, 
we draw the conclusions in Section \ref{sec:conclu}.

\subsection*{Notation and conventions}

We use the mostly minus signature for the metric tensor $(+---)$. We define the Riemann tensor as:
\begin{equation*}
    R^{\mu}_{\ \nu\rho\sigma}=\partial_{\sigma}\Gamma^{\mu}_{\ \nu\rho}-\partial_{\rho}
    \Gamma^{\mu}_{\ \nu\sigma}+\Gamma^{\mu}_{\ \sigma\alpha}\Gamma^{\alpha}_{\ \nu\rho}-
    \Gamma^{\mu}_{\ \rho\alpha}\Gamma^{\alpha}_{\ \nu\sigma}
\end{equation*}
and the Ricci tensor as $R_{\mu\nu}=R^{\alpha}_{\ \mu\alpha\nu}$. Our choice is represented as $(---)$ in the convention 
table of Misner, Thorne and Wheeler. The totally anti-symmetric tensor $E_{\mu\nu\rho\sigma}$ is defined as $\sqrt{-g}
\epsilon_{\mu\nu\rho\sigma}$ with $\epsilon$ the Levi-Civita symbol and $\epsilon^{0123}=+1$. 
The transverse projector to the four-velocity $u$ is $\Delta^{\mu\nu} = g^{\mu\nu}- u^\mu u^\nu$.

\section{Global thermodynamic equilibrium in curved space-time}
\label{sec:gte}

In the general covariant formulation of quantum statistical mechanics the state of a quantum system at global thermodynamic 
equilibrium is encoded in the following density operator \cite{Zubarev:1979afm, VANWEERT1982133, Becattini:2016stj}:
\begin{equation}\label{densop}
    \wrho =\frac{1}{Z}\exp{\left(-\int_{\Sigma}\di \Sigma_{\mu}\wT^{\mu\nu}\beta_\nu\right)}
\end{equation}
where $\Sigma$ is a space-like hypersurface belonging to a foliation of the space-time, and $\beta_\nu$ is the inverse 
temperature four-vector, or simply the four-temperature, with:
$$
T=1/\sqrt{\beta^2}
$$
being the proper or comoving temperature, measured by a thermometer moving with four-velocity $u= \beta/\sqrt{\beta^2}$.
A Killing vector can be suitably normalized in asymptotically flat space-times; alternatively, if there is a set of 
parameters $a_1,\ldots,a_n$ such that in the limit $a_i \to 0$ the metric becomes flat, one can normalize the Killing
vector so as to:
$$
  \beta = \frac{1}{T_0} \xi \qquad \lim_{a_i \to 0} \xi^\mu \xi^\nu g_{\mu\nu}(a_i) = 1 
$$
The dimensional parameter $T_0$ then becomes the {\em global temperature} and the relation between proper and global temperature reads:
$$
   T = \frac{T_0}{\sqrt{\xi^2}}
$$
which is the so-called Tolman relation.

At global equilibrium the operator $\wrho$ must be time-independent; in the language of general relativity, it should be 
independent of the spacelike hypersurface $\Sigma$. Provided that specific boundary conditions are fulfilled, the Gauss 
theorem implies that a necessary and sufficient condition for this to happen is that $\beta$ is a Killing vector:
\begin{equation}\label{killeq}
    \nabla_\mu\beta_\nu+\nabla_\nu\beta_\mu = 0
\end{equation}
Wherever $\beta$ is timelike, the calculation of thermal averages of local operators yields finite results and we can 
consider observers moving along its orbits. While the symmetric part of the gradient of $\beta_\mu$ is zero due 
to the Killing equation \eqref{killeq}, the antisymmetric part defines a rank two antisymmetric tensor $\varpi_{\mu\nu}$, 
called thermal vorticity. This tensor can be decomposed as the sum of a longitudinal (or electric) and a transverse (or magnetic) 
part using the time-like four-vector $u$:
\begin{equation}\label{thvort}
    \varpi_{\mu\nu} = \nabla_{[\nu}\beta_{\mu]} = \alpha_\mu u_\nu - \alpha_\nu u_\mu + 
    E_{\mu\nu\rho\sigma} w^\rho u^\sigma
\end{equation}
where
$$
E_{\mu\nu\rho\sigma} = \sqrt{-g} \, \epsilon_{\mu\nu\rho\sigma}
$$
and:
\begin{equation}\label{accthvort}
    \alpha_\mu = u^\nu\nabla_\nu\beta_\mu, \qquad w_\mu = 
    \frac{1}{2} E_{\mu\rho\sigma\nu}\nabla^\rho\beta^\sigma u^\nu
\end{equation}
At global equilibrium, $\alpha_\mu$ and $w_\mu$ can be linked to their kinematic analogs, acceleration and vorticity, 
through:
\begin{equation}\label{accthvort2}
    A_\mu=u^\nu\nabla_\nu u_\mu = \frac{1}{\sqrt{\beta^2}}\alpha_\mu, \qquad \omega_\mu=\frac{1}{2}
    E_{\mu\nu\rho\sigma}u^\nu\nabla^\rho u^\sigma = \frac{1}{\sqrt{\beta^2}}w_\mu
\end{equation}
When both acceleration and vorticity four-vectors are non-vanishing, it is useful to introduce a fourth vector:
\begin{equation}\label{lvector}
    l^\mu = E^{\mu\nu\rho\sigma}u_\nu A_\rho\omega_\sigma \qquad l^2=-A^2\omega^2+(A\cdot\omega)^2
\end{equation}
in order to complete a tetrad $\{u,A,\omega,l\}$. The vectors of this tetrad are orthogonal to each other by construction, 
except $A$ and $\omega$. The tetrad can become degenerate when $A$ and $\omega$ have the same direction. In this case we 
have two preferred directions, identified by the four-velocity and the four-acceleration and the tensor structure in the 
2D orthogonal space-like subspace is isotropic, that is it can be described simply by $g^{\mu\nu}-u^\mu u^\nu-\hat{A}^\mu\hat{A}^\nu$.

In principle, the mean value of a local operator $\wO(x)$ is defined as $\Tr(\wrho \, \wO(x))$, but this is usually
divergent and needs renormalization. A convenient way of renormalizing it is to define a vacuum-subtracted 
expectation value as:
\be\label{vacsub}
   \langle \wO(x) \rangle \equiv \Tr(\wrho \, \wO(x)) - \bra{0} \wO(x) \ket{0}
\ee
which is finite and yet it depends on the choice of the vacuum state $\ket{0}$. Vacuum states are naturally associated to 
the lowest lying states of the operator in the exponent of eq. \eqref{densop}, that is:
\be\label{vacuum}
\ket{0}_\beta = \arg \min_{\ket{\Phi} \in {\cal H}} 
\left[\bra{\Phi} \int_{\Sigma}\di \Sigma_{\mu}\wT^{\mu\nu}\beta_\nu \ket{\Phi} \right] 
\ee
hence to the Killing vector $\beta$ as apparent from the subscript in $\ket{0}_\beta$; in eq. \eqref{vacuum} $\cal H$ is the Hilbert space of quantum states. The so-called Minkowski vacuum $\ket{0_M}$, for instance, is obtained by taking
as Killing vector $\beta = 1/T_0(1,{\bf 0})$ with $1/T_0$ constant. If the Killing vector in eq. \eqref{vacuum} is proportional 
to the one in the density operator \eqref{densop}, then the vacuum-subtracted value in \eqref{vacsub} vanishes in the 
vacuum state, otherwise it will not. An example of the latter case is when we take a Killing vector with non-vanishing derivative
in Minkowski space-time (see below) and retain $\ket{0_M}$ in the equation \eqref{vacsub}. In this case, we have for
quadratic operators in the free fields:
$$
\langle \wO(x) \rangle = \Tr(\wrho \, \wO(x)) - \bra{0_M} \wO(x) \ket{0_M} =  \Tr(\wrho \, :\wO(x):)
$$
where the colon stands for normal ordering. 

The vacuum-subtracted expectation value, whatever the vacuum state is, is a {\em functional} of the classical functions associated 
to the density operator \eqref{densop}, namely the metric tensor $g$, the Killing vector field $\beta$ and the field $n(x)$ normal
to the hypersurface $\Sigma$:
$$
\langle \wO(x) \rangle = \langle \wO(x) \rangle[\beta,g,n]
$$
At global thermodynamic equilibrium, as discussed above, the dependence on the hypersurface is dropped, hence
the above value will have a functional dependence on $\beta$ and $g$ only. If $\beta$ and $g$ are analytic functions in $x$, 
we can replace the functional dependence with a dependence on an infinite set of arguments, the values of the functions 
$\beta,g$ and all of their derivatives in $x$. This allows us to transform the functional dependence on a dependence on 
all the gradients. In symbols:
\begin{equation*}
    \langle \wO(x) \rangle =  \langle \wO(x) \rangle [\beta,g] =  \langle \wO(x) \rangle(\beta(x),\partial\beta(x),\partial\partial\beta(x),\ldots,g(x),\partial g(x),\partial\partial g(x),\ldots)
\end{equation*}
By suitably choosing inertial coordinates around the space-time point $x$, it is possible to make the first order derivatives 
of the metric tensor vanishing and express all the arguments above in terms of non-vanishing tensors and covariant derivatives 
thereof:
\begin{equation}\label{expectO}
    \langle \wO(x) \rangle =  \langle \wO(x) \rangle(\beta(x),\nabla\beta(x),
    g(x),R,\nabla R(x),\ldots)
\end{equation}
where $R$ stands for the Riemann tensor. At global equilibrium \eqref{killeq}, the higher orders covariant derivatives of the 
Killing vector field do not appear because of the well-known relation:
\begin{equation*}
    \nabla_\mu\nabla_\nu\beta_\sigma = R^\lambda_{\ \ \mu\nu\sigma}\beta_\lambda
\end{equation*}
%

\section{Universality conjecture and analytic distillation}
\label{sec:conj}

For a given microscopic quantum field theory, an important question is whether the function \eqref{expectO} on the right 
hand side is universal or rather depends on the specific space-time under consideration. Indeed, in principle, the expectation
value of a local operator depends on the density operator \eqref{densop} and the integral in the exponent, even if it was independent 
of the space-like hypersurface $\Sigma$, may well depend on non-local properties. Amongst them, the global geometric properties of the
space-time and the boundary conditions enforced to ensure the independence of the density operator on the hypersurface $\Sigma$.
In this paper, we will show that, for some important space-times and for the stress-energy tensor, the analytic part of the 
function \eqref{expectO} in $\nabla\beta,R,\nabla R,\ldots$ about $\nabla\beta,R,\nabla R,\ldots=0$ is the same independently 
of the space-time, whereas the terms depending on global properties and boundary conditions are non-analytic in those variables. 
We will prove it explicitly for special space-time such as Minkowski, anti-de Sitter, de Sitter and the Closed Einstein Universe. 

Such regularity will lead us to conjecture that this is a universal feature, i.e. applying to all local operators (not just
the stress-energy tensor), all quantum fields and all possible space-times \footnote{Incidentally, this conjecture implies
that the analytic part of the solutions found for global equilibrium with rotation in Minkowski space-time with different boundary 
conditions is universal; in other words, boundary conditions enforced on the fields at the radius $R < 1/\omega_0$, where $\omega_0$
is the angular velocity, only introduce non-analytic corrections in $\omega_0$.}.

One of the key points in the above statements is the definition of an analytic part. Loosely speaking, the analytic part
of a function is obtained by carrying out a power series expansion of its arguments, considered as complex variables, starting 
from some point. For instance, for the function \eqref{expectO}, expanding in all the arguments $\nabla\beta,R,\nabla R,\ldots$ starting 
from $\nabla\beta,R,\nabla R,\ldots = 0$. Schematically:
\be\label{coeff}
 \langle \wO(x) \rangle(\beta(x),\nabla\beta(x), g(x),R,\nabla R(x),\ldots) = C_0
 + C_1 \nabla \beta(x) + C_2 R(x) + C_3 \nabla \beta(x) \nabla \beta(x) + \ldots
\ee
where the dependence of all the coefficients $C_i$ on $\beta(x),g(x)$ is understood. Notably, the above expansion corresponds 
to what is commonly known as a gradient expansion in relativistic hydrodynamics \cite{Baier:2007ix,Bhattacharyya:2007vjd,Jensen:2012jh}, 
which is generally an asymptotic one \cite{Heller:2013fn, Heller:2015dha, Grozdanov:2019kge} in its arguments; we note in passing 
that, for a gradient expansion, linear terms in $\nabla\beta$ are first order while linear terms in $R$ are second order as they 
involve second-order derivatives of the metric tensor. The thermodynamic coefficients $C_i$ of course depend on the underlying quantum 
field theory; concerning their possible dependence on the space-time, it is either tacitly assumed that they do not depend \cite{Khakimov:2023emy} 
altogether or the question is simply ignored because the main goal of these calculations is to calculate coefficients in some
specific space-time. It is precisely this kind of dependence that we are tackling in this work.

\subsection{Analytic distillation}

A thorough definition of the analytic part of a function of a complex variable requires some care because it is known that the same
function may have different asymptotic power series in different sectors of the complex plane, which is the so-called 
Stokes phenomenon. In order to illustrate the relevance of this point, it is useful to present a specific example. 
Suppose we deal with a scalar operator $\wO(x)$ and, in a particular space-time we have a leading order asymptotic expansion of the kind:
$$
 \langle \wO(x) \rangle \sim a + b R(x) + \ldots
$$
$R(x)$ being the curvature scalar and $a,b$ numerical coefficients. If, in this particular space-time, curvature scalar $R(x)$ is 
definite positive, we may have an ambiguity in singling out the universal analytic part because the above expression could be obtained 
from a more general:
$$
 \langle \wO(x) \rangle \sim a + b_1 R(x) + b_2 | R(x) | \ldots
$$
with $b_1+b_2 = b$; in this case, the universal analytic part includes only the first term, while the second is manifestly non-analytic. 
Indeed, the second term can be discarded only by studying the asymptotic expansion of the function $\langle \wO(x) \rangle$ in different 
sectors of the complex plane considering $R(x)$ as a complex variable.

In ref. \cite{Becattini:2020qol}, the definition of an analytic part of a complex function was proposed under the name of 
{\em analytic distillate}.
\begin{definition}\label{distillate}
{\em Let $f(z)$ be a function on a domain $D$ of the complex plane and $z_0 \in \bar D$ a point 
where the function may not be analytic. Suppose that asymptotic\footnote{We denote asymptotic
equality with the symbol $\sim$.} power series of $f(z)$ in $z-z_0$ 
exist in subsets $D_i \subset D$ such that $\cup_i D_i = D$: 
$$
  f(z) \sim \sum_{n} a^{(i)}_n (z-z_0)^n 
$$
where $n$ can take integer negative values. If the series formed with the common coefficients 
in the various subsets restricted to $n \ge 0$ has a positive radius of convergence, the analytic 
function defined by this power series is called analytic distillate of $f(z)$ in $z_0$ and it 
is denoted by $\dist_{z_0} f(z)$.}
\end{definition}
For the function \eqref{expectO}, with $\nabla\beta,R,\nabla R,\ldots$ complexified variables, the above definition makes it 
precise what should be retained in the so-called gradient expansion. 

For reasons that will become clear later, it is convenient to extend the definition of distillation by replacing, in its 
formulation, asymptotic power series with a more general asymptotic representation of the function, which is known as a transseries. A 
transseries (the reader is referred to specialized literature) is an infinite sum of terms, multiplied by scalar coefficients, 
which are not just powers but more general expressions possibly including real powers, logarithms and exponentials, called 
transmonomials. For instance:
$$
z^\alpha \exp[-1/z^\beta] \log z     \qquad \alpha,\beta \in {\mathbb R}
$$
are possible transmonomials. The set of allowed transmonomials is called group of transmonomials and the minimal group
is $\{ z^n \} \; \; n \in {\mathbb N}$. The definition of a distillate is then modified as follows:
\begin{definition}\label{distillate2}
{\em Let $f(z)$ be a function on a domain $D$ of the complex plane and $z_0 \in \bar D$ a point 
where the function may not be analytic. Suppose that the function $f(z)$, for a given group of transmonomials
at least including integer powers, has a transseries expansion in $z-z_0$ in subsets $D_i \subset D$ such 
that $\cup_i D_i = D$. If the series formed with the common coefficients in the various subsets restricted 
to integer positive powers of $(z-z_0)^n$ has a positive radius of convergence, the analytic function defined 
by this power series is called analytic distillate of $f(z)$ in $z_0$ and it is denoted by $\dist_{z_0} f(z)$.}
\end{definition}
For a simple example of distillation, consider the function:
$$
   f(z) = z + \e^{-1/z}
$$
which is already in the form of a transseries. Following the definition \ref{distillate2}, we have:
$$
   \dist_{z=0} f(z) = z
$$
It is important to note that the distillate does not depend on the initial group of transmonomials used for
the asymptotic representation of the function. Indeed, if we allow a wider class of transmonomials, the additional 
terms would be cancelled by distillation anyway.

According to the definition \ref{distillate2}, distillation is a linear operation, that is:
$$
\dist_{z=z_0} \left[ \lambda_1 F_1(z) + \lambda_2 F_2(z) \right] = \lambda_1 \dist_{z=z_0} F_1(z) + \lambda_2 
\dist_{z=z_0} F_2(z)
$$
On the other hand, in general, the distillate of the product of two functions is not the product of distillates:
$$
 \dist_{z=z_0} F_1(z) F_2(z) \ne \dist_{z=z_0} F_1(z) \dist_{z=z_0} F_2(z)
$$
the simplest counter-example being $F_1(z)=1/z,F_2(z)=z$ in $z=0$. Nevertheless, under additional hypotheses the
equality holds.
\begin{theorem} \label{thm1}
{\em Let $F(z)$ be an analytic function in $z_0$ and $G(z)$ a complex valued function with transseries expansions
about $z_0$ and no pole in $z_0$. In this case:
$$
 \dist_{z=z_0} F(z) G(z) = \dist_{z=z_0} F(z) \, \dist_{z=z_0} G(z) = F(z) \dist_{z=z_0} G(z)
$$
}
\end{theorem}
\begin{proof}
Under the assumption of the theorem one can write:
$$
  F(z) = \sum_{n \ge 0} a_n (z-z_0)^n
$$
and 
$$
  G(z) = \sum_{\alpha \in {\mathbb{R-N^-}}} b_\alpha (z-z_0)^\alpha + R(z-z_0)
$$
where $R(z-z_0)$ is a sum of transmonomials including non-trivial combinations of log-exp transmonomials of $(z-z_0)$, i.e. including 
factors like $\exp[-c_k/(z-z_0)^k]$ or $\log(z-z_0)$ and powers; however, the expansion of $G(z)$ 
does not feature terms proportional to $(z-z_0)^{-n}$ with $n$ positive integer. 
According to the definition of distillate, we have:
$$
 \dist_{z=z_0} F(z) = F(z) \qquad \dist_{z=z_0} G(z) = \sum_{\alpha \in {\mathbb{N}}} b_\alpha (z-z_0)^\alpha
$$
and:
$$
\dist_{z=z_0} F(z) G(z) = \dist_{z=z_0} \left[ 
\sum_{n \ge 0 \; \alpha \in {\mathbb{R-N^-}}} a_n b_\alpha (z-z_0)^{n+\alpha}  + \sum_{n \ge 0} a_n (z-z_0)^n R(z-z_0) \right]
$$
The second series on the right hand side has vanishing distillate, while:
$$
\dist_{z=z_0} \left[ \sum_{n \ge 0 \; \alpha \in {\mathbb{R-N^-}}} a_n b_\alpha (z-z_0)^{n+\alpha} \right]
= \left[ \sum_{n \ge 0 \; \alpha \in {\mathbb{N}}} a_n b_\alpha (z-z_0)^{n+\alpha} \right] 
= \dist_{z=z_0} F(z) \dist_{z=z_0} G(z)
$$
which concludes the proof.
\end{proof}

\subsection{Asymptotics of harmonic series}
\label{harmseries}

Most functions that we will be dealing with are \textit{harmonic series}, defined as \cite{FLAJOLET19953}:
\begin{equation*}
    F(x) = \sum_{k=1}^{\infty}\lambda_kf(\mu_k x)
\end{equation*}
where the coefficients $\lambda_k$ are referred to as amplitudes, the $\mu_k > 0$ are called frequencies and the function 
$f(x)$ is denoted as base function. Every harmonic series has a corresponding Dirichlet series, defined as:
\begin{equation*}
 \Xi(s)=\sum_{k=1}^{\infty}\lambda_k\mu_k^{-s}   .
\end{equation*}
Dirichlet series have a half-plane of absolute convergence, for example $\text{Re}(s)>\sigma_{abs}$, and a half-plane of 
simple convergence, for example $\text{Re}(s)>\sigma_{simple}$ with $\sigma_{simple}\geq\sigma_{abs}$.

In the literature several methods are proposed to calculate the asymptotic expansion of harmonic series,
but all can be derived from the Mellin summation formula \cite{FLAJOLET19953}, which is based on the Mellin 
transform. Given a function $f(x)$ that is Lebesgue integrable on the real axis, its Mellin transform is 
defined as:
\begin{equation*}
    f^{*}(s)=\int_{0}^{\infty} \di x \; f(x) \, x^{s-1}
\end{equation*}
The largest open strip of the complex plane where this integral converges is called the fundamental strip of $f^{*}(s)$. 
If $f(x)$ is $\mathcal{O}(x^{-A})$ for $x\to0^{+}$ and $\mathcal{O}(x^{-B})$ for $x\to\infty$ with $B>A$ then the 
fundamental strip is the region $A<\text{Re}(s)<B$.
There are two main reasons why asymptotic methods based on the Mellin transform are well suited for harmonic series. 
The first one is the so-called "separation" property, while the second is the existence of a correspondence between 
the asymptotic expansion of $f(x)$ and the poles of $f^{*}(s)$. The separation property is the fact that the Mellin 
transform of harmonic series can be written as the product of the associated Dirichlet series and the transform of 
the base function:
\begin{equation*}
    F^{*}(s)=\Xi(s)f^{*}(s).
\end{equation*}
The correspondence can be understood with a simple example. Given $f^{*}(s)$ we can obtain $f(x)$ by inverting the 
transform:
\begin{equation*}
    f(x) = \frac{1}{2\pi i}\int_{c-i\infty}^{c+i\infty} \di s \; f^{*}(s) \, x^{-s}
\end{equation*}
where ${\rm Re} s=c$ identifies a vertical line within the fundamental strip. Suppose $f^{*}(s)$ can be analytically 
continued outside the fundamental strip and that it has a pole for $s=-N$ with residue $f_{\rm res}$. We can evaluate the 
inverse transform by considering a counterclockwise rectangular contour delimited by the vertical lines 
$\text{Re}(s)=-(N+\epsilon)$ and $\text{Re}(s)=c$. If the integral on the horizontal lines at complex infinity vanishes 
then using the residue theorem we can write:
\begin{equation*}
    f(x) = f_{\rm res}x^{N} - \frac{1}{2\pi i}\int_{-N-\varepsilon-\ii \infty}^{-N-\varepsilon+\ii \infty} 
    \hspace{-0.7 cm} \di s \; 
    f^{*}(s) \, x^{-s} = f_{\rm res}x^{N} + \mathcal{O}(x^{N+\varepsilon})
\end{equation*}
By taking bigger values for $\epsilon$ one can include sub-leading contributions coming from the poles in the region $s<-N$. Eventually one can send $\epsilon\to\infty$ and obtain a full power-series asymptotic expansion in the limit $x\to 0^+$.

In order to present the main theorem on the asymptotic expansions of harmonic series we need to introduce the concept of 
functions of \textit{fast decrease} and of \textit{moderate growth}.

A function $f(s)$ is said to be of fast decrease in a closed strip of the complex plane $A\leq\text{Re}(s)\leq B$ if
\begin{equation*}
    f(s) = \mathcal{O}(|s|^{-\alpha})
\end{equation*}
for every $\alpha>0$ as $|s|\to\infty$.

A function $f(s)$ is said to be of moderate growth in a closed strip of the complex plane $A \leq \text{Re}(s) \leq B$ if 
\begin{equation*}
    f(s) = \mathcal{O}(|s|^{\alpha})
\end{equation*}
for some $\alpha>0$ as $|s|\to\infty$.

In practice these definitions tell us that a fast decrease function decays faster than any negative power of $|s|$ while a
moderate growth function increases at most as fast as a positive power of $|s|$. 
We can now state the Mellin summation formula \cite{FLAJOLET19953}:
\begin{theorem}\label{thmfl}
Let $F(x)$ be an harmonic series and a continuous function of the variable $x$ on the real line $(0,\infty)$. If:
\begin{enumerate}[label=\roman*)]
    \item there is a non-empty intersection between the fundamental strip of the Mellin transform of the base function
    $f^{*}(s)$ and the half-plane of absolute convergence of $\Xi(s)$,
    \item $f^{*}(s)$ is meromorphic in $\mathbb{C}$ and of \textit{fast decrease}, 
    \item $\Xi(s)$ is meromorphic in $\mathbb{C}$ and of \textit{moderate growth},
\end{enumerate}
then $F(x)$ admits an asymptotic expansion for $x\rightarrow0^{+}$ that has the following form:
\begin{equation*}
    F(x)\sim\sum_{\text{Re} (s)<c}\text{Res}\left(f^{*}(s)\Xi(s)x^{-s}\right),
\end{equation*}
where $c$ lies in the intersection between the fundamental strip of $f^{*}(s)$ and the half-plane of absolute 
convergence of $\Xi(s)$.
\end{theorem}
From this general theorem two corollaries follow relating the asymptotic expansion of special harmonic series, with
$\lambda_k = k$ to the asymptotic expansion of the base function.
These results have been proved in ref. \cite{ZagierAppendixTM} for functions of real variables, and later discussed and
extended to complex functions in ref. \cite{Becattini:2020qol} and \cite{Palermo:2021hlf} in the context of analytic
distillation.
\begin{corollary}\label{coro1}
{\em Let $F$ be a $C^\infty$ complex valued function in a domain of the complex plane and suppose
that $F$ has the following asymptotic power series in $z=0$:
$$
  F(z) \sim \sum_{k=-M}^\infty A_k z^k
$$
with $M$ a positive integer and such that $F(z)-A_{-1}/z = \mathcal{O}(1/|z|^{1+\varepsilon})$ for $|z| \to +\infty$ 
with $\varepsilon > 0$. Then, the function defined by the series:
\begin{equation*}
   G(z) = \sum_{n=1}^\infty F(nz)
\end{equation*}
has the asymptotic expansion for $|z| \to 0^+$:
\begin{equation*}
   G(z) \sim \frac{1}{z}\left[ - A_{-1} (\log z - \ii \arg z) + I_{\arg z} \right] 
   + \sum_{\substack{n=-M \\ n \ne -1}}^\infty A_n \zeta(-n) z^n ,
\end{equation*}  
where
\begin{equation*}
  I_{\arg z} = \int_\Gamma \di w \; \left( F(w)-\sum_{n=-M}^{-2} A_n w^n - A_{-1} \frac{\e^{-w\exp[-\ii \arg z]}}{w} \right).
\end{equation*}
and $\Gamma$ is the line in the complex plane from $|z|=0$ to $+\infty$ with angle $\arg z$.}
\end{corollary}
According to the definition \ref{distillate2}, the distillate of the function $G(z)$ in the above corollary is:
$$
  \dist_{z=0} G(z) = \sum_{n=0}^\infty A_n \zeta(-n) z^n
$$
Note that the integral $I_{\arg z}$ for $\arg z = 0$ reads:
$$
  I_{0} = \int_0^{\infty} \di x \; \left( F(x)-\sum_{n=-M}^{-2} A_n x^n - A_{-1} \frac{\e^{-x}}{x} \right).
$$
It is worth noting that if $A_{-1}=0$ and the integrand function of $I_{\arg z}$ has no poles in the 
sector between $\arg z=0$ and $\arg z = \pi$, then $I_{\arg z} = I_0$ according to the assumed asymptotic 
conditions for $|z| \to +\infty$ for any $0 \leq \arg z \leq \pi$. If, on the other hand, the integrand 
function has a pole, the integral will differ as soon as $\arg z$ exceeds the argument of the pole, giving rise 
to a different integral coefficient in the asymptotic expansion (i.e. the Stokes phenomenon).
\begin{corollary}\label{coro2}
{\em Let $F(z)$ be a $C^\infty$ complex valued function in a domain of the complex plane and suppose
that $F$ has the following asymptotic power series in $z=0$ 
$$
  F(z) \sim \sum_{k=-M}^\infty A_k z^k
$$
with $M$ a positive integer and let $F$ be $o(1/|z|)$ when $|z| \to +\infty$ in the real axis.
Let $G(z)$ be the function defined by the series:
\begin{equation}\label{G}
   G(z) = \sum_{n=1}^\infty (-1)^{n+1} F(nz).
\end{equation}
The asymptotic power series of $G(z)$ for $|z| \to 0^+$, is given by:
\begin{equation}
\label{eq:AsymptG}
   G(z) \sim \sum_{n=-M}^\infty A_n \eta(-n) z^n ,
\end{equation}
where $\eta$ is the Dirichlet function:
\begin{equation*}
\eta(s)=\sum_{k=1}^{\infty}\frac{(-1)^{k+1}}{k^s}=(1-2^{1-s})\zeta(s).
\end{equation*}}
\end{corollary}
Note that both the Riemann Zeta function and the Dirichlet Eta function vanish when evaluated on negative even integers. 
When the asymptotic series of the base function contains only even powers, so that $A_{2k+1}=0$, 
the asymptotic expansion of the complete harmonic series will contain a finite number of terms, \textit{i.e.} 
it will be a polynomial.

\section{Stress-energy tensor of a massless scalar field in Minkowski space-time}
\label{sec:mink}

For a real massless free scalar field at global equilibrium, all the coefficients $C_i$ in eq. \eqref{coeff} up to second order 
in the gradient expansion of the stress-energy tensor have been calculated using functional methods \cite{Moore:2012tc} or, 
in flat space-time, the density operator method \cite{Buzzegoli:2017cqy}. It has been shown that these coefficients can be 
written in terms of just two quantities, that appear at the level of the equilibrium generating functional, called thermodynamic 
susceptibilities \cite{Kovtun:2018dvd}. Finally, in ref. \cite{Becattini:2020qol} for the scalar field and in ref. 
\cite{Palermo:2021hlf} for the Dirac field, the exact expression of the stress-energy tensor (both canonical and improved) 
normalized with subtraction of the Minkowski vacuum, was obtained in flat space-time and shown to correspond to a truncated
power series in $\nabla \beta$ at the fourth order. This result was obtained by means of analytic distillation and continuation
from imaginary acceleration and angular velocity. In the case of pure acceleration the results precisely match the exact results 
obtained by solving the field equations in curvilinear coordinates whereas in the case of pure rotation, the obtained expression
refers to floating boundary condition, which makes the exponent of the density operator \eqref{densop} unbounded from below.
In the case of global equilibrium in the presence of both acceleration and vorticity the result for the improved stress-energy 
tensor reads \cite{Becattini:2020qol}:
\begin{equation*}
    \langle:\wT^{\mu\nu}:\rangle=\rho' u^{\mu}u^{\nu}-p'\Delta^{\mu\nu}+W'\omega^{\mu}\omega^{\nu}+\mathcal{A}'A^{\mu}A^{\nu}+G^{l}l^{\mu}l^{\nu}+G'(u^{\mu}l^{\nu}+u^{\nu}l^{\mu})
\end{equation*}
where $\Delta^{\mu\nu} = g^{\mu\nu} - u^\mu u^\nu$ and:
\begin{equation*}
\begin{split}
    &\rho'=\frac{\pi^2}{30\beta^4}-\frac{1}{36\beta^2}\omega^2-\frac{1}{144\pi^2}\omega^4-\frac{1}{480\pi^2}A^4-\frac{1}{45\pi^2}A^2\omega^2, \\
    &p'=\frac{\pi^2}{90\beta^4}-\frac{1}{36\beta^2}\omega^2-\frac{1}{144\pi^2}\omega^4-\frac{1}{1440\pi^2}A^4,\\
    &W'=-\frac{1}{18\beta^2}-\frac{1}{72\pi^2}\omega^2-\frac{1}{40\pi^2}A^2,\\
    &\mathcal{A}'=\frac{1}{360\pi^2}\omega^2,\\
    &G^{l}=-\frac{1}{1440\pi^2},\\
    &G'=\frac{1}{18\beta^2}-\frac{1}{360\pi^2}\omega^2-\frac{1}{360\pi^2}A^2.
\end{split}
\end{equation*}
For the canonical stress-energy tensor, which corresponds to minimal coupling to gravity in curved space-time, the tensor 
structure is the same and the coefficients are given by:
\begin{equation*}
\begin{split}
    & \rho'=\frac{\pi^2}{30\beta^4}-\frac{1}{12\beta^2}\omega^2-\frac{1}{12\beta^2}A^2-\frac{1}{48\pi^2}\omega^4-
    \frac{11}{480\pi^2}A^4-\frac{61}{720\pi^2}A^2\omega^2,\\
    &p'=\frac{\pi^2}{90\beta^4}+\frac{1}{18\beta^2}A^2+\frac{19}{1440\pi^2}A^4,\\
    &W'=-\frac{1}{12\beta^2}-\frac{1}{48\pi^2}\omega^2-\frac{29}{360\pi^2}A^2,\\
    &\mathcal{A}'=\frac{1}{12\beta^2}+\frac{1}{360\pi^2}\omega^2+\frac{1}{48\pi^2}A^2,\\
    &G^{l}=\frac{1}{240\pi^2},\\
    &G'=\frac{1}{36\beta^2}-\frac{1}{240\pi^2}\omega^2-\frac{7}{720\pi^2}A^2.
\end{split}
\end{equation*}
In the derivation of refs. \cite{Becattini:2020qol,Palermo:2021hlf} the coefficient $G^l$ is indeed redundant
as $l^\mu l^\nu$ can be expressed as a linear combination of the metric tensor $g$ and the symmetric tensor $A^\mu A^\nu$
and $\omega^\mu \omega^\nu$:
\begin{equation}
    l^\mu l^\nu= -A^2\omega^2\Delta^{\mu\nu}-\omega^2A^\mu A^\nu-A^2\omega^\mu\omega^\nu
    \label{decomp}
\end{equation}
Therefore, the two expressions above can be recast as:
\begin{equation}\label{minkset}
    \langle:\wT^{\mu\nu}:\rangle=\rho u^{\mu}u^{\nu}-p\Delta^{\mu\nu}+W\omega^{\mu}\omega^{\nu}+\mathcal{A}A^{\mu}A^{\nu}+G(u^{\mu}l^{\nu}+u^{\nu}l^{\mu})
\end{equation}
with the following coefficients for $\xi=1/6$:
\begin{equation}\label{minkfunc}
\begin{split}
    &\rho=\frac{\pi^2}{30\beta^4}-\frac{1}{36\beta^2}\omega^2-\frac{1}{144\pi^2}\omega^4-\frac{1}{480\pi^2}A^4-\frac{1}{45\pi^2}A^2\omega^2, \\
    &p=\frac{\pi^2}{90\beta^4}-\frac{1}{36\beta^2}\omega^2-\frac{1}{144\pi^2}\omega^4-\frac{1}{1440\pi^2}A^4-\frac{1}{1440\pi^2}A^2\omega^2,\\
    &W=-\frac{1}{18\beta^2}-\frac{1}{72\pi^2}\omega^2-\frac{7}{288\pi^2}A^2,\\
    &\mathcal{A}=\frac{1}{480\pi^2}\omega^2,\\
    &G=\frac{1}{18\beta^2}-\frac{1}{360\pi^2}\omega^2-\frac{1}{360\pi^2}A^2.
\end{split}
\end{equation}
and for $\xi=0$:
\begin{equation}\label{minkfunc2}
\begin{split}
    & \rho=\frac{\pi^2}{30\beta^4}-\frac{1}{12\beta^2}\omega^2-\frac{1}{12\beta^2}A^2-\frac{1}{48\pi^2}
    \omega^4-\frac{11}{480\pi^2}A^4-\frac{61}{720\pi^2}A^2\omega^2,\\
    &p=\frac{\pi^2}{90\beta^4}+\frac{1}{18\beta^2}A^2+\frac{19}{1440\pi^2}A^4+\frac{1}{240\pi^2}A^2\omega^2,\\
    &W=-\frac{1}{12\beta^2}-\frac{1}{48\pi^2}\omega^2-\frac{61}{720\pi^2}A^2,\\
    &\mathcal{A}=\frac{1}{12\beta^2}-\frac{1}{720\pi^2}\omega^2+\frac{1}{48\pi^2}A^2,\\
    &G=\frac{1}{36\beta^2}-\frac{1}{240\pi^2}\omega^2-\frac{7}{720\pi^2}A^2.
\end{split}
\end{equation}
These expressions are, however, incomplete as they miss terms depending on $A \cdot \omega$, which have been
shown to appear in the full expression of the stress-energy tensor of the massless Dirac free field \cite{Palermo:2021hlf}.
Hence, the equation \eqref{minkset} applies only if $ A \cdot \omega = 0$, as well as the decomposition \eqref{decomp}.

The universality conjecture discussed in Section \ref{sec:conj} states that the coefficients in \eqref{minkset} are the
same for the analytic distillate of the stress-energy tensor of the same quantum field in different space-times. 
We will show that this is the case for three different space-times, as has been mentioned: anti-de Sitter, de Sitter 
and the Closed Einstein Universe, for which exact solutions are available for the stress-energy tensor at global thermodynamic
equilibrium.

\section{Stress-energy tensor in Anti-de Sitter space-time}
\label{sec:ads}

\subsection{Anti-de Sitter geometry}

Anti-de Sitter space-time, AdS for short, is a negative constant curvature space-time and it is a solution of the Einstein field equation 
with a negative cosmological constant. It is maximally symmetric, namely it possesses the maximum allowed number of independent Killing 
vectors, ten in four dimensions. AdS can be defined as a four-dimensional hypersurface embedded in a five-dimensional flat 
space-time with signature $(+,-,-,-,+)$ with equation:
\begin{equation}\label{adsembedding}
    X_{0}^{2}-X_{1}^{2}-X_{2}^{2}-X_{3}^{2}+X_{4}^{2}=-a^{2},
\end{equation}
where $a > 0$ is referred to as the AdS radius; the inverse of the AdS radius is the curvature $\kappa=1/a$. 
The topology of AdS is given by $S^{1}\times \mathbb{R}^{3}$. The embedding coordinates can be parametrized as:
\begin{equation}
    \begin{aligned}
    X_{0}&=a\cos{\tau}\sec{\rho}, \  
    X_{1}=a\tan{\rho}\cos{\theta}, \
    X_{2}=a\tan{\rho}\cos{\varphi}\sin{\theta}, \\
    X_{3}&=a\tan{\rho}\sin{\varphi}\sin{\theta}, \
    X_{4}=a\sin{\tau}\sec{\rho},
    \end{aligned}
    \label{embedding-AdS}
\end{equation}
where $-\pi<\tau\leq\pi$, $0\leq\rho<\frac{\pi}{2}$, $0\leq\theta\leq\pi$ and $0\leq\phi\leq2\pi$. In this set of coordinates 
the metric is given by:
\begin{equation}
    \di s^{2}=a^{2}\sec^{2}{\rho}\left[\di \tau^{2}-\di \rho^{2}-\sin^{2}{\rho}(\di \theta^{2}+
    \sin^{2}{\theta}\, \di \varphi^{2})\right].
    \label{AdSmetric1}
\end{equation}
The periodicity of the time coordinate $\tau$ signals the presence of closed timelike curves. These curves are not physical 
and have to be removed through the universal cover of AdS. This corresponds to considering consecutive copies of AdS, allowing 
the time coordinate to vary on the entire real axis. $S^{1}$ is then replaced by $\mathbb{R}^1$ so the topology of the universal 
covering is $\mathbb{R}^4$. From now on we will always be referring to the universal cover. 

AdS is not globally hyperbolic because of the presence of a timelike boundary, located at $\rho=\pi/2$. Light rays manage to 
reach the boundary in a finite time, implying that information can be gained or lost through spatial infinity. Even
though not globally hyperbolic, AdS can be foliated into spacelike hypersurfaces, so the thermodynamic problem with the density
operator \eqref{densop} is well posed.

We can introduce a new set of coordinates by setting $t=a\tau$ and $r=a\tan{\rho}$. The resulting line element can be written in 
terms of $\kappa \equiv 1/a$ as:
\begin{equation}\label{AdSmetric}
    \di s^{2} = (1+r^2\kappa^2)\di t^{2}-(1+r^2\kappa^2)^{-1}\di r^{2}-r^{2}\left[\di \theta^{2}+\sin^{2}{\theta}\di \varphi^{2}\right].
\end{equation}
and it is very useful because its limit for $\kappa \to 0$ is manifestly the Minkowski metric in spherical coordinates.

In AdS there is a globally timelike Killing vector, which in the above coordinates reads:
\begin{equation}\label{KillingAdS}
    \beta^{\mu}=\frac{1}{T_{0}}(1,0,0,0) \implies \beta^{2}=\frac{1}{T_{0}^{2}}(1+r^2\kappa^2)=\beta_{0}^{2}(1+r^2\kappa^2).
\end{equation}
where $T_0$ is the global temperature (see discussion in Section \ref{sec:gte}) and we defined:
$$
 \beta_0 \equiv \frac{1}{T_0}
$$
The Killing vector in eq. \eqref{KillingAdS} has zero vorticity $\omega_\mu$ (see definition \eqref{accthvort2}) but non-vanishing acceleration:
\begin{equation}\label{A Ads}
    A_{\mu}=-\frac{r\kappa^2}{1+r^2\kappa^2}(0,1,0,0) \implies A^{2}=-\frac{r^2\kappa^4}{1+r^2\kappa^2}.
\end{equation}
The Riemann tensor is greatly constrained by the maximal symmetry and reads:
\begin{equation}\label{Riemannads}
    R_{\mu\nu\rho\sigma}=-\kappa^2(g_{\mu\rho}g_{\nu\sigma}-g_{\mu\sigma}g_{\nu\rho}) \; 
\end{equation}
so that the Ricci tensor reads:
\be\label{ricads}
R_{\mu\nu}=-3\kappa^2 g_{\mu\nu}
\ee
and the curvature scalar:
\be\label{curvads}
R =-12 \kappa^2 
\ee
%

\subsection{The stress-energy tensor at global equilibrium in AdS}

An exact renormalized value of the stress-energy tensor at finite temperature in AdS, with the Killing vector in \eqref{KillingAdS}
for the massless scalar real, conformally coupled, free field was derived in ref. \cite{Allen:1986ty}. The calculation was performed for 
both Dirichlet and Neumann boundary conditions using point splitting and Hadamard renormalization. The result can be written as:
\begin{equation}\label{AllenTmunu}
    \langle \wT^{\mu\nu}\rangle_{\rm ren}=\frac{\kappa^4}{960\pi^{2}}g^{\mu\nu}+ T^{\mu\nu}_{\text{bulk}} \pm T^{\mu\nu}_\partial
\end{equation}
where:
\begin{equation}\label{bulkset}
    T_{\mu\nu \; \rm bulk} = \frac{1}{\left( 6\pi ^{2}\right) }\left(\frac{\kappa^2}{1+\kappa^2r^{2}}\right)^{2}f_{3}
    \left( \kappa\beta_0\right) (-g_{\mu\nu}+4u_\mu u_\nu)
\end{equation}
and:
\begin{align}\label{boundset}
    T_{\mu\nu \; \partial} =  & \ \frac{\kappa^3}{\left( 8\pi ^{2}r\right)}\Biggl\{ \Biggl[ +\frac{1}{6}\left(\frac{1-\kappa^2r^2}{\kappa^2r^{2}}\right)S_{0}\left( \kappa\beta_0,r\right)-\frac{1}{3\kappa r}C_{1}\left( \kappa\beta_0,r\right)+ \nonumber
    \\
    &-\frac{2}{3}\left(\frac{1}{1+\kappa^2r^{2}}\right) S_{2}\left( \kappa\beta_0,r \right) \Biggl] g_{\mu\nu}
    +\Biggl[\frac{1}{6}\left(3-\frac{1}{\kappa^2r^{2}}\right)S_{0}\left( \kappa\beta_0,r\right)+ \nonumber
    \\
    &+\frac{1}{\kappa r}\left(1-\frac{2}{3}\frac{1}{1+\kappa^2r^{2}}\right)C_{1}\left( \kappa\beta_0,r\right)
    +\frac{2}{1+\kappa^2r^{2}}S_{2}\left( \kappa\beta_0,r\right)\Biggl]u_{\mu}u_{\nu}+ \nonumber
    \\
    &+\Biggl[\frac{1}{6}\left(3\frac{1+\kappa^2r^2}{\kappa^2r^{2}}-4\right)S_{0}\left( \kappa\beta_0,r\right)
    +\frac{1}{\kappa r}\left(\frac{2}{3}\frac{\kappa^2r^{2}}{1+\kappa^2r^{2}}-1\right)C_{1}\left( \kappa\beta_0,r\right)
    + \nonumber \\
    &-\frac{2}{3}\frac{1}{1+\kappa^2r^{2}}S_{2}\left( \kappa\beta_0,r\right)\Biggl]\left(\frac{1+r^2\kappa^2}{r^2\kappa^4}\right)A_{\mu}A_{\nu}\Biggl\}
\end{align}
$u_\mu$ being the the four-velocity $u= \beta/\sqrt{\beta^2}$, and $A_\mu$ the acceleration four-vector
\footnote{In the ref. \cite{Allen:1986ty} the expression \eqref{boundset} contains a small typographical error, notably there is 
a missing $(\kappa r)^{-1}$ factor as pointed out in \cite{Morley:2023exv}.}. The plus sign in eq. \eqref{AllenTmunu} refers to 
Dirichlet boundary conditions, while the minus sign to Neumann boundary conditions, so the term $T^{\mu\nu}_\partial$ entirely 
depends on the boundary conditions. The series in eq. \eqref{bulkset} and \eqref{boundset} read:
\begin{align}\label{series1}
  &  f_{3}\left( \kappa\beta_0\right)=\sum_{n=1}^{\infty}\frac{n^{3}}{\e^{n\kappa\beta_0}-1} \nonumber \\
  &  S_{m}\left( \kappa\beta_0,r\right)=\sum_{n=1}^{\infty}\frac{n^{m}(-1)^{n}}{\e^{n\kappa\beta_0}-1}\sin{(2n\arctan{(\kappa r)})} \\
  &  C_{m}\left( \kappa\beta_0,r\right)=\sum_{n=1}^{\infty}\frac{n^{m}(-1)^{n}}{\e^{n\kappa\beta_0}-1}\cos{(2n\arctan{(\kappa r)})} \nonumber .
\end{align}
where the range of the $\arctan$ function is $[0,\pi/2)$ being the argument positive. If transparent boundary conditions are
chosen, the boundary term \eqref{boundset} vanishes; however this choice is only possible for conformal coupling \cite{Avis:1977yn}.
The three terms in eq. \eqref{AllenTmunu} are independently conserved and, with the exception of the first one, traceless. 

As it is clear from eq. \eqref{AllenTmunu}, the renormalized stress-energy tensor includes a term which is proportional 
to $g_{\mu\nu}$; this is the well-known conformal anomaly and the point-splitting renormalization method is responsible 
for its appearance. The remaining two terms in eq. \eqref{AllenTmunu} vanish in the limit $T_0 \to 0$, so the 
vacuum-subtracted stress-energy tensor reads (see definition \eqref{vacsub}):
\be\label{vsubset}
   \langle \wT^{\mu\nu} \rangle \equiv \Tr(\wrho \, \wT^{\mu\nu} ) - \bra{0} \wT^{\mu\nu} \ket{0}
   = \langle \wT^{\mu\nu} \rangle_{\rm ren} - {\rm Trace\;anomaly} = T^{\mu\nu}_{\text{bulk}} \pm T^{\mu\nu}_\partial
\ee
where $\ket{0}$ is the vacuum associated to the Killing vector in eq. \eqref{KillingAdS}. 

More recently, the exact stress-energy tensor in global AdS was evaluated, both for conformal and minimal coupling, for 
a real free scalar field with Dirichlet boundary conditions \cite{Ambrus:2018olh}. In the quoted expressions, bulk and 
boundary contributions are entangled. The vacuum-subtracted stress-energy tensor is:
\begin{equation}\label{AmbrusTmunu}
    \langle \hat{T}_{\mu\nu} \rangle = \rho \ u_\mu u_\nu -p \ \Delta_{\mu\nu} + \mathcal{A} \ A_\mu A_\nu
\end{equation}
and, for the conformal coupling ($\xi=1/6$) the scalar fields in \eqref{AmbrusTmunu} read:
\begin{equation*}
\begin{split}
&    \rho = \frac{\kappa^4}{192\pi^2(1+r^2\kappa^2)^2}\sum_{n=1}^{\infty}\frac{R(n\kappa\beta_0)}
{(1-r^2\kappa^2+(1+r^2\kappa^2)\cosh{(n\kappa\beta_0)})^3\sinh^4{(n\kappa\beta_0/2)}}  \\
&    p = \frac{\kappa^4}{192\pi^2(1+r^2\kappa^2)^2}\sum_{n=1}^{\infty}\frac{P_{1/6}(n\kappa\beta_0)}
{(1-r^2\kappa^2+(1+r^2\kappa^2)\cosh{(n\kappa\beta_0)})^3\sinh^4{(n\kappa\beta_0/2)}}   \\
& \mathcal{A} = \frac{2\kappa^2(1+r^2\kappa^2)}{3\pi^2}\sum_{n=1}^{\infty}\frac{\sinh^2{(n\kappa\beta_0/2)}}
{(1-r^2\kappa^2+(1+r^2\kappa^2)\cosh{(n\kappa\beta_0)})^3}
\end{split}
\end{equation*}
where:
\begin{align*}
    R(n\kappa\beta_0) =& \left(\kappa ^2 r^2+1\right)^2 \left(\kappa ^2 r^2+21\right) 
    \cosh (3 n \kappa\beta_0 )-6 \left(\kappa ^6 r^6+5 \kappa ^4 r^4-5 \kappa ^2 r^2-9\right) 
    \cosh (2 n \kappa\beta_0 ) + \nonumber \\
    &+3 \left(5 \kappa ^6 r^6-29 \kappa ^4 r^4-25 \kappa ^2 r^2+41\right) \cosh (n \kappa\beta_0 )-10 \kappa ^6 r^6+94 
    \kappa ^4 r^4+2 \kappa ^2 r^2+90
\end{align*}
and 
\begin{align*}
    P_{1/6}(n\kappa\beta_0)=&-\left(\kappa ^2 r^2-7\right) \left(\kappa ^2 r^2+1\right)^2 \cosh (3 n \kappa \beta_0)+6 
    \left(\kappa ^4 r^4+3\right) \left(\kappa ^2 r^2+1\right) \cosh (2 n \kappa\beta_0 )+\nonumber\\
    &-\left(15 \kappa ^6 r^6+69 \kappa ^4 r^4+45 \kappa ^2 r^2-41\right) \cosh (n \kappa \beta_0)+2 
    \left(5 \kappa ^6 r^6+29 \kappa ^4 r^4+7 \kappa ^2 r^2+15\right).
\end{align*}
For the minimal coupling ($\xi=0$) the coefficients are given by:
\begin{equation*}
\begin{split}
&    \rho = \frac{3\kappa^4}{8\pi^2(1+\kappa^2r^2)^2} \sum_{n=1}^\infty 
\frac{1}{(1-r^2\kappa^2+(1+r^2\kappa^2)\cosh{(n\kappa\beta_0)})\sinh^4{(n\kappa\beta_0/2)}}  \\
&    p = \frac{\kappa^4}{16\pi^2(1+\kappa^2r^2)^2}\sum_{n=1}^\infty\frac{P_0(n\kappa\beta_0)}
{(1-r^2\kappa^2+(1+r^2\kappa^2)\cosh{(n\kappa\beta_0)})^3\sinh^4{(n\kappa\beta_0/2)}}  \\
& \mathcal{A} = \frac{\kappa^2}{2\pi^2(1+r^2\kappa^2)} \sum_{n=1}^{\infty} 
 \frac{3 \left(\kappa ^2 r^2+1\right) \cosh (n \kappa\beta_0 )-3 \kappa ^2 r^2-1}{(1-r^2\kappa^2+(1+r^2\kappa^2)\cosh{(n\kappa\beta_0)})^3 \sinh^2{(n\kappa\beta_0/2)}}
\end{split}
\end{equation*}
where:
\begin{align*}
    P_0(n\kappa\beta_0) =&+ \left(3-3 \kappa ^4 r^4\right) \cosh (2 n \kappa\beta_0 )+4 
    \left(3 \kappa ^4 r^4+2 \kappa ^2 r^2+1\right) \cosh (n \kappa\beta_0 )+\nonumber \\
    &-9 \kappa ^4 r^4-8 \kappa ^2 r^2+1 \ ,
\end{align*}
%

\subsection{Analytic distillate of the stress-energy tensor in AdS: conformal coupling}
\label{DistAdsCC}

We are now going to apply the analytic distillation technique to the foregoing expressions of the stress-energy tensor.

Before tackling the AdS case a general discussion on how to perform analytic distillation of a stress-energy tensor is
in order. Let us generically denote by $D \beta$ and $D g$ the covariant derivatives of the four-temperature and of the metric, 
respectively. Since we are interested in an analytic dependence of the stress-energy tensor on $D \beta$ and $Dg$ in zero, in
principle one should first write down the stress-energy tensor $T_{\mu\nu}$ as a function of $D\beta,Dg$. However, this method 
might be inconvenient; the stress-energy tensor is given in terms of parameters, such as $\kappa$ in equations 
\cref{AllenTmunu,bulkset,boundset} and re-expressing in terms of $D\beta,Dg$ might be a complicated task. It is often more 
convenient to identify parameters which analytically depend on $D\beta,Dg$ without explicit coordinate dependence and apply 
distillation directly to them. Since the composition of analytic functions is analytic, the final goal would be achieved. In 
symbols, if $T^{\mu\nu}$ is given as a function of some parameters $\alpha_i$:
$$
  T^{\mu\nu}(\alpha_i)
$$
and if:
$$
  \alpha_i = F_i(D\beta,Dg)  
$$
are analytic functions in $D\beta=0,Dg=0$, then performing an analytic distillation of $T^{\mu\nu}$ with respect to the $\alpha_i$'s 
in the appropriate points would suffice to ensure the resulting stress-energy tensor to be an analytic function of $D\beta,Dg$ in $D\beta,Dg=0$. 

A crucial requirement of distillation is the independence of the coordinate system. The stress-energy tensor is 
given in contravariant or covariant components on a particular coordinate frame and one would like the familiar 
transformation relations to apply to distillates as well, namely:
$$
 \dist T^{\prime\mu\nu} = \frac{\partial x'^\mu}{\partial x^\alpha}\frac{\partial x'^\nu}{\partial x^\beta} \dist T^{\alpha\beta}
 \qquad\qquad
 \dist T_{\mu\nu} = g_{\alpha\mu} g_{\alpha\nu} \dist T^{\alpha\beta}
$$
For this purpose, one has to keep in mind theorem \ref{thm1}; if the maps $x'(x)$ are analytic with respect to the considered
parameters $\alpha_i$ and the metric tensor as well, the above relations are satisfied provided that the stress-energy tensor,
for either its covariant or contravariant components, have no pole for $\alpha_i = F(0,0)$. In the first place, one should 
therefore choose a coordinate frame such that the metric tensor and its inverse are analytic functions of the parameters $\alpha_i$.

In order to simplify the distillation process, it is useful to decompose the stress-energy tensor onto tetrads built with vectors
formed with the tetrad $\{u,A,\omega,l\}$ (see equation \eqref{lvector}), which are analytic combinations of the 
covariant derivative of $\beta$, according to formulae \eqref{accthvort},\eqref{accthvort2}; a decomposition of this sort has been 
presented in \eqref{minkset} and \cref{AllenTmunu,bulkset,boundset}. If the scalar coefficients of the combinations of tetrad 
vectors meet the requirements of the theorem \ref{thm1}, i.e. they do not have poles in the distillation parameters $\alpha_i$ 
in $\alpha_i=F(0,0)$, then it suffices to distillate the scalar functions. 

In the case at hand, i.e. the stress-energy tensor in equations \cref{AllenTmunu,bulkset,boundset} there is only one parameter
$\kappa$, so, according to the foregoing discussion, one should first write the latter as an analytic function of $D\beta,Dg$. 
According to the equations \eqref{curvads}, $\kappa^2$ can be expressed as an analytic function of $R$ \footnote{As has been
mentioned, it is important that parameters such as $\kappa$ can be expressed as a function of the derivatives of $\beta$ and $g$ 
without further coordinate dependence. For instance, $\kappa^2$ cannot be expressed as a function of $\beta^2$ without an explicit 
additional dependence on $r$.}:
$$
  \kappa^2 = - \frac{R}{12}
$$
Thus, $\kappa^2$ can be chosen as the only variable to which distillation can be applied. It should be pointed 
out though, that $\kappa^2$ can be equally written as an analytic function of other variables. For instance, by using equations \eqref{KillingAdS},\eqref{A Ads} and \eqref{ricads}:
$$
\kappa^2 = A^2 \frac{\beta^2}{\beta_0^2 - \beta^2}   \qquad \qquad \kappa^2 = -\frac{1}{3} R_{\mu\nu} u^\mu u^\nu
$$
This ambiguity leads to a problem in singling out the dependence of the stress-energy tensor on the curvature tensors and
the survival of possible non-analytic expressions such as, for instance:
$$
  \sqrt{R (R_{\mu\nu} u^\mu u^\nu)}
$$
after analytic distillation in $\kappa^2$. Such an ambiguity can be fixed only by studying the corresponding problem in a 
space-time where the degeneracy of the above expressions is removed; we will see that this happens in the closed Einstein 
universe in Section \ref{sec:ceu}.

All the scalar functions appearing in the equations \cref{AllenTmunu,bulkset,boundset} also fulfill the requirement of 
having no poles for $\kappa^2=0$; this is not manifest but it can be seen once the asymptotic expansions have been set up. 
Indeed, all the scalar coefficients are given in terms of series in eq. \eqref{series1} with $\kappa$ as the main parameter. 
Since $\kappa = \sqrt{\kappa^2}$ is not an analytic function of $\kappa^2$, there seems to be a difficulty 
in distilling with respect to $\kappa^2$. However, with the definition of distillate \ref{distillate2}, it suffices to 
determine the asymptotic expansion in the variable $\kappa$ by using the theorems in Section \ref{sec:conj}; the replacement
$\kappa \to (\kappa^2)^{1/2}$ will provide a trans-series of the same function in $\kappa^2$, to which analytic distillation
can be readily applied. Furthermore, no terms in the series \eqref{series1} can be expressed as an analytic function of the
second variable to be distilled, namely $A^2$. Hence, we can proceed to the asymptotic expansion in $\kappa$ and subsequent 
distillation of the series in $\kappa^2$ in eq. \eqref{series1} without affecting the second variable $A^2$.

We can first work out the asymptotic expansion of the series $f_3(\kappa\beta_0)$ in the equation \eqref{bulkset}. Since:
\be\label{f3series}
   f_3(\kappa\beta_0) = \frac{1}{\kappa^3} \sum_{n=1}^{\infty} \frac{n^3\kappa^3}{\e^{n\kappa\beta_0}-1}
\ee
the domain of existence of this function is ${\rm Re}\,\kappa > 0$, implying that for $\kappa^2$ the domain
is the full complex plane except the negative real semi-axis. Taking into account that:
$$
   \frac{\kappa^3}{\e^{\kappa\beta_0}-1} = \sum_{k=0}^\infty \frac{B_k}{k!} \beta_0^{k-1} \kappa^{k+2} 
   = \sum_{n=2}^\infty \frac{B_{n-2}}{(n-2)!} \beta_0^{n-3} \kappa^{n} 
$$
we can apply the Corollary \ref{coro1} and obtain the following relation:
\be\label{asyf3}
    \sum_{n=1}^{\infty} \frac{n^3\kappa^3}{\e^{n\kappa\beta_0}-1} \sim \frac{I_{\arg \kappa}}{\kappa} + 
    \zeta(-2)B_0 \frac{\kappa^2}{\beta_0} + \sum_{n=3}^\infty \zeta(-n) \frac{B_{n-2}}{(n-2)!} \beta_0^{n-3} \kappa^{n} = 
    \frac{I_{\arg \kappa}}{\kappa} -\sum_{n=3}^\infty \frac{\zeta(-n)\zeta(-n+3)}{(n-3)!} \beta_0^{n-3}\kappa^{n}
\ee
where:
$$
  I_{\arg \kappa} = \int_\Gamma \di w \; \frac{w^3}{\e^{w\beta_0}-1} 
$$
with $-\pi/2 < \arg \kappa < \pi/2$. 
Now, in the equation \eqref{asyf3} all terms in the rightmost series vanish except $n=3$ because $\zeta(-2l)=0$ for
all positive integers. At the same time, the integral $I_{\arg \kappa}$ can be calculated on the real line because
the integrand function has no poles in the domain $-\pi/2 < \arg w < \pi/2$ and decays rapidly for $w \to \infty$, 
so:
\begin{align}\label{asymf3}
  \sum_{n=1}^{\infty} \frac{n^3\kappa^3}{\e^{n\kappa\beta_0}-1} &\sim \frac{1}{\kappa} 
  \int_0^\infty \di x \; \frac{x^3}{\e^{x\beta_0}-1} -\zeta(-3)\zeta(0) \kappa^3 \nonumber \\
  & = 6 \zeta(4) \frac{1}{\kappa \beta_0^4} - \zeta(-3)\zeta(0) \kappa^3 
  = \frac{\pi^4}{15} \frac{1}{\kappa \beta_0^4} - \frac{1}{240} \kappa^3
\end{align}
We are now in a position to calculate the distillate with respect to $\kappa^2$ in $\kappa^2=0$ of the bulk term of 
the stress-energy tensor in equation \eqref{bulkset}. By using equations \eqref{f3series},\eqref{asymf3} and the 
theorem \ref{thm1} for the distillation of products of functions, we get:
\begin{align}\label{bulkdistillate}
    & \dist_{\kappa^2=0} T_{\mu\nu \; \rm bulk} =  \dist_{\kappa^2=0} \left[ f_3(\kappa\beta_0) 
    \left(\frac{\kappa^2}{1+\kappa^2r^2}\right)^2\left( -g_{\mu\nu}+4u_{\mu}u_{\nu} \right)\right] \nonumber \\
    & = \frac{1}{6\pi^2} \left(\frac{1}{1+\kappa^2r^2}\right)^2 \left[ \dist_{\kappa^2=0} \kappa \sum_{n=1}^{\infty} \frac{n^3\kappa^3}{\e^{n\kappa\beta_0}-1} \right]\left( -g_{\mu\nu}+4u_{\mu}u_{\nu} \right)
 \nonumber \\
 & =  \left(\frac{1}{1+\kappa^2r^2}\right)^2\left( \frac{\pi^2}{90} \frac{1}{\beta_0^4} - \frac{1}{1440 \pi^2} \kappa^4 \right)
 \left( -g_{\mu\nu}+4u_{\mu}u_{\nu} \right)
\end{align}

We now move on to distill the boundary term \eqref{boundset} comprising the series $S_0$, $C_1$ and $S_2$. According
to their expressions \eqref{series1}, they are not harmonic series in general, yet we can use a trick in order to perform 
the asymptotic expansion using the Corollaries \ref{coro1},\ref{coro2}. Define, with $m \geq 0$:
\begin{equation}\label{emfunc}
    E_{m}(\kappa\beta_0,r)=\sum_{n=1}^{\infty}\frac{n^{m}(-1)^{n}}{\e^{n\beta_{0}\kappa}-1}\e^{2 \ii n\arctan(r\kappa)}
\end{equation}
which is such that
$$
    C_{m}\left( \beta_{0}\kappa,r\right)=\text{Re} (E_{m}\left( \beta_{0}\kappa,r\right)), \qquad
    S_{m}\left( \beta_{0}\kappa,r\right)=\text{Im} (E_{m}\left( \beta_{0}\kappa,r\right)).
$$
Again the domain of the functions $E_m(\kappa\beta_0,r)$ is ${\rm Re} \, \kappa >0$ and $\kappa^2 \ne {\rm real\;negative}$.
The trick, analogous to that used in ref. \cite{Becattini:2020qol}, is the introduction of an auxiliary parameter $y$
in the function \eqref{emfunc} such that the extended function $E_m(\kappa\beta_0,r,y)$ is an harmonic series and, 
in the limit $y\to\kappa$, it precisely becomes the right hand side of the equation \eqref{emfunc}. The extended
function reads:
\begin{equation}\label{extended}
    E_{m}\left(\beta_{0}\kappa,r, y \right)\equiv -\left(\frac{1}{\kappa\beta_{0}}\right)^{m}
    \sum_{n=1}^{\infty}(-1)^{n+1}\frac{(n\beta_{0}\kappa)^{m}}{\e^{n\beta_{0}\kappa}-1}\exp\left[i\frac{2n\beta_{0}\kappa}
    {y\beta_{0}}\arctan(ry)\right].
\end{equation}
Since the series \eqref{extended} is uniformly convergent in its domain of existence (${\rm Re} \, \kappa > 0$,
and $y$ such that ${\rm Im} [(\kappa/y) \arctan (ry)] > 0$) we can interchange the limit and the series and
obtain:
$$
  \lim_{y \to \kappa} E_{m}\left(\beta_{0}\kappa,r, y \right) = E_{m}(\kappa\beta_0,r,\kappa)
$$
It is then possible to obtain the asymptotic series of the function in eq. \eqref{emfunc} by taking the limit 
$y \to \kappa$ of the asymptotic power 
series in $\kappa$ of the function \eqref{extended} if the coefficients of the various powers of $\kappa$ are analytic 
functions of $y$ in $y=0$. Indeed, in this case, taking the limit $y \to \kappa$ will not introduce any singularity in 
$\kappa=0$ and the coefficients can then be expanded in a Taylor series about $\kappa =0$ without concern. In symbols, 
if in the asymptotic expansion of the function in \eqref{extended} we have:
$$
  E_{m}\left(\beta_{0}\kappa,r, y \right) \sim \sum_{n} a_n(y,r) \kappa^n + b_n(y,r)[{\rm other \; transmonomials]}
$$
and if the functions $a_n(y,r), b_n(y,r)$ are analytic in $y=0$, then, according to known properties of the
transseries:
$$
   E_{m}\left(\beta_{0}\kappa,r, \kappa \right) \sim \sum_{n} \lim_{y \to \kappa} a_n(y,r) \kappa^n
   + \lim_{y \to \kappa} b_n(y,r)[{\rm other \; transmonomials]}
$$
Thereafter, the functions $a_n(\kappa,r)$ can be expanded in $\kappa$ to obtain the desired final power series.

Let us now see how this trick works in some detail. The equation \eqref{extended} can be recast as:
\begin{equation}\label{recastem}
    E_{m}\left(\beta_{0}\kappa,r,y\right)=-\left(\frac{1}{\kappa \beta_{0}}\right)^{m}
    \sum_{n=1}^{\infty}(-1)^{n+1}F_{m}\left(nx\right),
\end{equation}
where $x=\kappa\beta_0$ and 
$$
    F_{m}(x)=\frac{x^m\e^{i\Phi(y)x}}{\e^x-1}, \qquad \Phi(y)=\frac{2}{y\beta_{0}}\arctan(ry).
$$
The function $F_{m}(x)$ can be written in turn as a power series whose coefficients are linked to Bernoulli 
polynomials:
$$
    F_{m}(x)=\frac{x^m\e^{i\Phi(y)x}}{\e^x-1}
    =\sum_{l=m-1}^{\infty}\frac{B_{l-m+1}\left(i\Phi(y))\right)}{(l-m+1)!}x^{l}.
$$
Using the Corollary \ref{coro2} for the function \eqref{recastem} and the above expression for $F_m(x)$ 
we obtain the following asymptotic expansion:
\begin{equation}\label{recastemasy}
   E_{m}\left(\beta_{0}\kappa,r,y\right) \sim-\left(\frac{1}{\kappa\beta_{0}}\right)^{m}\sum_{l=m-1}^{\infty}
   \frac{B_{l-m+1}\left(i\Phi(y)\right)}{(l-m+1)!}\eta(-l)x^{l}.
\end{equation}
where $\eta$ is the Dirichlet Eta function. Since all the coefficients of the powers of $x = \kappa \beta_0$ are analytic in $y$ 
at $y=0$ (the function $\Phi(y)=(1/y)\arctan(ry)$ is analytic in $y=0$ and $B_n(ic)$ are polynomials), we can safely take the 
limit $y \to \kappa$ term by term to obtain the asymptotic expansion of the original function \eqref{emfunc} for 
$\kappa \to 0$. 

We can now proceed to the calculation of $S_0$. In view of the above discussion, from the equation \eqref{recastemasy} 
and taking the limit within the series:
\begin{equation}\label{S0exp}
    S_{0}\left(\beta_{0}\kappa,r\right)\sim - \text{Im} \left(\sum_{l=-1}^{\infty}
    \lim_{y\rightarrow\kappa} \frac{B_{l+1}\left(i\Phi(y)\right)}{(l+1)!}\eta(-l)(\beta_{0}\kappa)^{l}\right) \; .
\end{equation}
Bernoulli polynomials are given by:
$$
    B_{n}\left(X\right)=\sum_{k=0}^{n}\binom{n}{k}B_{n-k} X^k,
$$
where $B_{n-k}$ are the Bernoulli numbers. Bernoulli numbers with odd index vanish except $B_1$, hence Bernoulli polynomials 
with even index $B_{2p}(X)$ only have one odd term in $X$:
$$
    B_{2p}(X)=-\frac{1}{2}(2p)X^{2p-1}+\sum_{m=0}^{p}A_{m}X^{2m}.
$$
Conversely, Bernoulli polynomials with an odd index $B_{2p+1}(x)$ only have one even term in $x$:
$$
    B_{2p+1}(X)=-\frac{1}{2}(2p+1)X^{2p}+\sum_{m=0}^{p}A_{m+1}X^{2m+1}.
$$
Now, the Dirichlet Eta function vanishes for negative even integers so the only non-vanishing terms in eq. \eqref{S0exp} 
are those with $j=2p+1$, $p$ being an integer. Since:
$$
    \text{Im} B_{2p+2}(i\Phi(y))=(-1)^{p+1}(p+1)\Phi(y)^{2p+1}.
$$
and taking the limit $y \to \kappa$ in the equation \eqref{S0exp}, keeping in mind the definition of $\Phi(y)$, we 
get:
\be\label{S0exp2}
    S_{0}\left(\beta_{0}\kappa,r\right)\sim-\frac{1}{\beta_{0}\kappa}\arctan(r\kappa)+
    \sum_{l=0}^{\infty}(-1)^{l}\frac{(l+1)}{(2l+2)!}\eta(-(2l+1))(2\arctan(r\kappa))^{2l+1}
\ee
Now:
\begin{equation}\label{etaBern}
    \eta(-(2p+1))=\frac{2^{2p+2}-1}{2p+2}B_{2p+2},
\end{equation}
hence for the second term on the right hand side of \eqref{S0exp2} we get:
\begin{align*}
    &\sum_{l=0}^{\infty}(-1)^{l}\frac{1}{2}\frac{2^{2l+2}-1}{(2l+2)!}B_{2l+2}(2\arctan(r\kappa))^{2l+1} \\
    =&\frac{1}{4}\sum_{L=1}^{\infty}(-1)^{L-1}\frac{2^{2L}-1}{(2L)!}B_{2L}2^{2L}(\arctan(r\kappa))^{2L-1}=
    \frac{1}{4}\tan(\arctan(r\kappa)) = \frac{1}{4}\kappa r
\end{align*}
where we have used the known Taylor expansion of the tangent function. The final result for the asymptotic
expansion of the function \eqref{S0exp} is thus:
\begin{equation}\label{S0exp3}
    S_{0}\left(\beta_{0}\kappa,r\right)\sim-\frac{1}{\beta_{0}\kappa}\arctan(r\kappa)+\frac{1}{4}r\kappa.
\end{equation}

With a similar process, reported in Appendix \ref{app:boundaryads}, we obtain the asymptotic expansions for the 
other two series:
\begin{align}\label{C1S2series}
 &C_{1}\left(\beta_{0}\kappa,r\right)\sim-\frac{1}{2\beta_{0}\kappa}+\frac{1}{8}+\frac{1}{8}(r\kappa)^2 
 \nonumber \\
 &S_{2}\left(\beta_{0}\kappa,r\right)\sim-\frac{1}{8}r\kappa-\frac{1}{8}(r\kappa)^3.
\end{align}
The equations \eqref{S0exp3} and \eqref{C1S2series} can be combined to obtain the asymptotic expansion of
the boundary term in powers of $\kappa$ of the stress-energy tensor in the equation \eqref{boundset}:
\begin{align}\label{AdSasymptoticexpansion}
 T_{\mu\nu \, \partial} \sim&\frac{\kappa^2}{48\pi^2r\beta_{0}}\Biggl\{ \Biggl[\frac{1}{r\kappa}+\Biggl(1-\frac{1}{(r\kappa)^2}\Biggl)\arctan(r\kappa)\Biggl]g_{\mu\nu}+ \nonumber
    \\
    &+\Biggl[\Biggl(\frac{1}{(r\kappa)^2}-3\Biggl)\arctan(r\kappa)-\frac{1}{r\kappa}
    \Biggl(3-\frac{2}{(r\kappa)^2+1}\Biggl)\Biggl]u_{\mu}u_{\nu}+ \nonumber
    \\
    &+\Biggl[\Biggl(1-\frac{3}{(r\kappa)^2}\Biggl)\arctan(r\kappa)-\frac{1}{r\kappa}\Biggl(\frac{2(r\kappa)^2}{(r\kappa)^2+1}-3\Biggl)\Biggl]\left(\frac{1+r^2\kappa^2}{r^2\kappa^4}\right)A_{\mu}A_{\nu}\Biggl\}.
\end{align}
The asymptotic power series in $\kappa$ of the \eqref{AdSasymptoticexpansion} can now be easily obtained by 
expanding in a power series all the analytic functions involved. However, all terms of this expansion will have 
{\em odd} exponents in $\kappa$, meaning fractional exponents in $\kappa^2$. In symbols:
$$
 T_{\mu\nu \, \partial} \sim \sum_{n=1}^\infty a_n (\kappa^2)^{(2n+1)/2}
$$
Therefore, according to the definition of the distillate \ref{distillate2}, we achieve a remarkable conclusion:
\be\label{boundvanish}
\dist_{\kappa^2 = 0}  (T_{\mu\nu \, \partial}) = 0
\ee
that is the term of the stress-energy tensor which depends on the boundary conditions in AdS has no analytic part.
This conclusion is the first evidence of the proposed universality conjecture.

Combining the equation \eqref{AllenTmunu} with \eqref{bulkdistillate} and the equation \eqref{boundvanish} 
we finally obtain the analytic distillate of the vacuum-subtracted stress-energy tensor of the conformally coupled free
massless scalar field in AdS:
\be\label{distadsset}
 \dist_{\kappa^2 = 0} \langle \wT_{\mu\nu}\rangle =
 \left( \frac{\pi^2}{90} T_0^4-\frac{\kappa^4}{1440\pi^{2}}\right) \left(\frac{1}{1+\kappa^2r^2}\right)^2
 \left( -g_{\mu\nu}+4u_{\mu}u_{\nu} \right).
\ee
The same result can be obtained from the exact calculation in the form of eq. \eqref{AmbrusTmunu},
where the boundary contribution is entangled with the bulk expression. In this case all the coefficients are already 
in the form of a harmonic series and distillation can be worked out by using Corollary \ref{coro1}. This result 
matches what we would obtain by using the Page approximation for the stress-energy tensor of a conformally coupled scalar 
field \cite{Page:1982fm}.

The distillate \eqref{distadsset} has remarkable features: 
\begin{itemize}
    \item {} it has vanishing trace, just like the original vacuum-subtracted stress-energy tensor in
    eq. \eqref{vsubset};
    \item{} it is covariantly conserved, i.e. $\nabla_\mu \dist_{\kappa^2=0} \langle \wT^{\mu\nu} \rangle=0$;
    \item{} unlike the original vacuum-subtracted stress-energy tensor in \eqref{vsubset} it does not vanish
    at $T_0=0$ but at $T_0 = \kappa/2\pi$
\end{itemize}
The temperature $T_0=\kappa/2\pi$ is linked to the Unruh effect in AdS space-time, as will be discussed in Section
\ref{sec:unruh}. 

It is interesting to note that the boundary term in eq. \eqref{AdSasymptoticexpansion} which has been distilled away
includes terms which could not be neglected in an asymptotic expansion of the full expression of the stress-energy tensor. 
For instance, the pressure term proportional to $g_{\mu\nu}$ in eq. \eqref{AdSasymptoticexpansion} for small $\kappa$ 
can be approximated with the leading term:
$$
  p_\partial \simeq \frac{1}{48 \pi^2 \beta_0} \frac{4}{3} \kappa^3 = \frac{1}{36 \pi^2 \beta_0} \left(-\frac{R}{12}\right)^{3/2} 
$$
where $R$ is the curvature scalar. For small curvature, this correction is more important than the term $\propto \kappa^4$
in eq. \eqref{distadsset}. Yet, this correction is peculiar to AdS; for dS, as we will see, no correction proportional to 
the third power of the square root of curvature is present. In other words, while analytic terms are universal, the non-analytic 
ones are explicitly space-time dependent. Similar terms, characterized by a peculiar scaling given by $\beta_0^{-1}$, have been observed in the presence of a boundary \cite{Kennedy:1979ar}.

\subsection{Analytic distillate of the stress-energy tensor in AdS: minimal coupling}

The asymptotic expansion of the coefficients in \eqref{AmbrusTmunu} for minimal coupling can be obtained by using 
the same methods presented for the conformal coupling.
The proper energy density in eq. \eqref{AmbrusTmunu} can be written as:
\begin{equation*}
    \rho = \frac{3\kappa^4}{8\pi^2(1+\kappa^2r^2)^2}\sum_{n=1}^\infty
    \frac{1}{(1-r^2\kappa^2+(1+r^2\kappa^2)\cosh{(n\kappa\beta_0}))\sinh^4{(n\kappa\beta_0/2)}}
\end{equation*}
which is defined for ${\rm Re} \kappa \ne 0$. Like for the conformal coupling, we should first cast the above series 
in a form suitable for the application of the Corollary \ref{coro1}. We then define:
\begin{equation*}
    f_\rho(\kappa\beta_0,y)=\frac{1}{(1-r^2 y^2+(1+r^2 y^2)\cosh{(\kappa\beta_0}))\sinh^4{(\kappa\beta_0/2)}}
\end{equation*}
and study the series:
\begin{equation*}
    \rho = \frac{3\kappa^4}{8\pi^2(1+\kappa^2r^2)^2} \lim_{y \to \kappa} \sum_{n=1}^\infty f_\rho (n \kappa \beta_0,y)
\end{equation*}
taking advantage of uniform convergence in its domain of existence. The power series of $f_\rho(\kappa\beta_0,y)$ begins 
with a term proportional to $\kappa^{-4}$ and contains only even powers, hence, according to the Corollary \ref{coro1}, 
the asymptotic series of $\rho$ has a finite number of terms. For the proper energy density the following expansion
is eventually obtained:
\begin{equation*}
    \rho \sim \frac{3\kappa^4}{8\pi^2(1+r^2\kappa^2)^2} \left(\frac{8\pi^4}{90\kappa^4\beta_0^4}-
    \frac{\pi^2(5+3r^2\kappa^2)}{9\kappa^2\beta_0^2}+\frac{I_{\arg \kappa}}{\kappa\beta_0}-
    \frac{(71+105r^2\kappa^2+45r^4\kappa^4)}{180}\right)
\end{equation*}
where:
\begin{equation*}
    I_{\arg \kappa} = \lim_{y \to \kappa} \int_\Gamma \di w \; \left(\frac{1}{(1-r^2 y^2+(1+r^2y^2)\cosh{(w\beta_0)})
    \sinh^4{(w \beta_0/2)}}-\frac{8}{w^4\beta_0^4}+\frac{2(5+3r^2y^2)}{3w^2\beta_0^2}\right)
\end{equation*}
and $\Gamma$ is the line with fixed $\arg \kappa$. In principle, the value of the integral depends on the value of $\arg \kappa$
as the integrand function may have poles for given values of $y$. In the limit $y = \kappa$ it can be shown that the integrand 
function has non-trivial poles, so according to the definition of analytic distillation, the term involving $I_{\arg \kappa}$
should be discarded. Besides, whatever the value of the integral, the term $I_{\arg \kappa}/\kappa$ features odd powers of 
$\kappa$, which are removed by analytic distillation in $\kappa^2$ anyway. Altogether we get:
\begin{equation}\label{rhods}
    \dist_{\kappa^2=0} \rho = \frac{\pi^2}{30(1+r^2\kappa^2)^2\beta_0^{4}} - \frac{\kappa^2(5+3r^2\kappa^2)}{24(1+r^2\kappa^2)^2\beta_0^2} -\frac{\kappa^4(71+105r^2\kappa^2+45r^4\kappa^4)}{480\pi^2(1+r^2\kappa^2)^2}  
\end{equation}

The calculation of the other scalar functions proceeds in a similar fashion and their final distillates in $\kappa^2=0$
read:
\begin{align}\label{othersds}
    \dist_{\kappa^2=0} p &= \frac{\pi^2}{90(1+r^2\kappa^2)^2\beta_0^{4}} + \frac{\kappa^2(1-3r^2\kappa^2)}{72(1+r^2\kappa^2)^2\beta_0^2} +\frac{\kappa^4(139+255r^2\kappa^2+135r^4\kappa^4)}{1440\pi^2(1+r^2\kappa^2)^2} 
    \nonumber \\
  \dist_{\kappa^2=0} \mathcal{A} &= \left(\frac{1+r^2\kappa^2}{r^2\kappa^4}\right)\biggl[\frac{r^2\kappa^4}{12(1+r^2\kappa^2)^2\beta_0^2} 
  +\frac{r^2\kappa^6}{48\pi^2(1+r^2\kappa^2)^2}\biggl]    
\end{align}

The final expression of the stress-energy tensor in the minimal coupling thus involves more terms than in the conformal
coupling, being ${\cal A} \ne 0$ and with less simple functions of the radial coordinate. We will elaborate further on this 
expression in Section \ref{sec:uni}.
    
\section{Stress-energy tensor in de Sitter space-time}
\label{sec:ds}

We begin this Section with a brief recapitulation about de Sitter space-time. De Sitter space-time (dS) can be 
identified with the hypersurface:
$$
    X_{0}^{2}-X_{1}^{2}-X_{2}^{2}-X_{3}^{2}-X_{4}^{2}=-a^{2}
$$
embedded in a five dimensional flat space-time with one timelike direction. The embedding coordinates can be
parametrized as:
\begin{align*}
    X_{0}&=a\sinh{\frac{\tau}{a}}, \\
    X_{1}&=a\cosh{\frac{\tau}{a}}\sin{\rho}\cos{\theta}, \\
    X_{2}&=a\cosh{\frac{\tau}{a}}\sin{\rho}\sin{\theta}\sin{\varphi}, \\
    X_{3}&=a\cosh{\frac{\tau}{a}}\sin{\rho}\sin{\theta}\cos{\varphi}, \\
    X_{4}&=a\cosh{\frac{\tau}{a}}\cos{\rho},
\end{align*}
where $-\infty<\tau<\infty$, $0<\rho<\pi$, $0<\theta<\pi$ and $0<\varphi<2\pi$. The metric in these coordinates reads:
\begin{equation}\label{dSmetricGlobal}
    \di s^{2}=\di \tau^{2}-a^2\cosh^2{\frac{\tau}{a}}\left[\di \rho^2+\sin^2{\rho}\left(\di \theta^{2}+
    \sin^{2}{\theta}\, \di \varphi^{2}\right)\right].
\end{equation}
This set of coordinates fully covers dS space-time, thus it is called global. Note that the constant 
time hypersurfaces are spheres and they are Cauchy surfaces for dS, meaning that dS is globally hyperbolic and can
be foliated into Cauchy space-like hypersurfaces.
However this set of coordinates is clearly not static, as there is an explicit dependence on the time coordinate
of the metric tensor. Indeed, dS space-time does not have a globally timelike Killing vector. Nonetheless, it is
possible to use a different parametrization of the embedding coordinates:
\begin{align*}
    X_{0}&=\sqrt{a^2-r^2}\sinh{\frac{t}{a}}, \\  
    X_{1}&=r\cos{\theta}, \\
    X_{2}&=r\sin{\theta}\cos{\varphi}, \\
    X_{3}&=r\sin{\theta}\sin{\varphi}, \\
    X_{4}&=\sqrt{a^2-r^2}\cosh{\frac{t}{a}},
\end{align*}
where $0\leq r\leq a$. In this parametrization, denoting $\kappa = 1/a$ like in AdS case, the metric tensor reads:
\begin{equation}  \label{dSmetric}
    \di s^{2}=(1-r^2\kappa^2)\di t^{2}-(1-r^2\kappa^2)^{-1}\di r^{2}-r^{2}\left(\di \theta^{2}+\sin^{2}{\theta}\, 
    \di\varphi^{2}\right).
\end{equation}
and is manifestly static, i.e. in this region $\partial_t$ is a timelike Killing vector. However, this set of coordinates 
is not global and it only covers a sub-manifold of the de Sitter space-time, the so-called static patch. 
When $r=\kappa^{-1}$ the Killing vector becomes light-like, so a Killing horizon is present beyond which the vector
becomes spacelike. The global thermodynamical equilibrium described by:
\begin{equation}\label{KillingdS}
    \beta^{\mu}=\frac{1}{T_{0}}(1,0,0,0) \implies \beta^{2}=\frac{1}{T_{0}^{2}}(1-r^2\kappa^2)
\end{equation}
has vanishing vorticity $\omega_\mu$ and non-vanishing acceleration:
\begin{equation}
    A_{\mu}=\frac{r\kappa^2}{1-r^2\kappa^2}(0,1,0,0) \implies A^{2}=-\frac{r^2\kappa^4}{1-r^2\kappa^2}.
    \label{AccDS}
\end{equation}
Note that it is possible to recover the AdS global metric \eqref{AdSmetric} from the static patch metric in dS 
through the transformation $\kappa^2\to -\kappa^2$. The same is true for all the relevant scalars, such as $\beta^2$ 
and $A^2$ and curvature invariants.

Field quantization in dS can be carried out in the standard way and the presence of a null horizon does not require 
boundary conditions on the fields. However, regularity requirements limit the nature of quantum states that can be 
considered. For instance it is well known that the presence of a cosmological horizon is strongly linked to the 
emergence of thermodynamical properties, namely the Gibbons-Hawking temperature \cite{Gibbons:1977mu}, which will 
be discussed in a dedicated section.

Finite temperature quantum field theory in the static patch of dS has been considered in \cite{Fursaev:1993hm}. The 
exact stress-energy tensor for a conformally coupled scalar field has been calculated in \cite{Fursaev:1994yj} and
reads:
\begin{equation}\label{dsset}
\langle \wT_{\mu\nu}\rangle_{\rm ren} = \frac{\kappa^4}{960\pi^{2}}g_{\mu\nu}
 + \left( \frac{\pi^2}{90} T_0^4-\frac{\kappa^4}{1440\pi^{2}}\right) \left(\frac{1}{1-\kappa^2r^2}\right)^2
 \left( -g_{\mu\nu}+4u_{\mu}u_{\nu} \right).
\end{equation}
Thus, removing the trace anomaly just like in AdS case:
\begin{equation}\label{dsdistset}
\dist_{\kappa^2=0} \langle \wT_{\mu\nu}\rangle = 
 \left( \frac{\pi^2}{90} T_0^4-\frac{\kappa^4}{1440\pi^{2}}\right) \left(\frac{1}{1-\kappa^2r^2}\right)^2
 \left( -g_{\mu\nu}+4u_{\mu}u_{\nu} \right).
\end{equation}
which corresponds to the analytic distillate \eqref{distadsset} in AdS with the transformation $\kappa^2 \to -\kappa^2$,
thus confirming the universality conjecture. 
Indeed, the equation \eqref{dsdistset} could have been expected based on the equation \eqref{distadsset}. Since the metric
tensor in dS is obtained from that in AdS by changing $\kappa^2 \to -\kappa^2$, and that $\kappa^2$ is proportional
to the curvature tensors, the analytic part of the stress-energy tensor in dS must be obtained by analytically
continuing the analytic part (i.e. the distillate) of the stress-energy tensor in AdS, because the continuation
of the series \eqref{series1} appearing in the exact solution is impossible. The nice feature is that the analytic 
continuation yields not just the distillate in dS, in fact the exact result in dS \eqref{dsset}. In this respect, the 
situation is the same as that described in ref. \cite{Becattini:2020qol} for global equilibrium with acceleration in 
Minkowski space-time; while an analytic continuation from imaginary to real acceleration was impossible for the 
full exact expression, the analytic continuation of the distillate matches the exact result. 

The \eqref{dsdistset} has the same features as the expression of the distilled vacuum-subtracted stress-energy tensor
in AdS \eqref{distadsset}: vanishing trace, covariantly conserved and vanishing for $T_0 = \kappa/2\pi$. We will
discuss more about this in Section \ref{sec:uni}. It also matches the result obtained using Page approximation 
\cite{Page:1982fm}, which is exact in dS.

\section{Stress-energy tensor in the Closed Einstein Universe with vorticity}
\label{sec:ceu}

In the foregoing Sections we focused on global equilibrium states with vorticity-free Killing vector fields. In 
curvilinear coordinates, in which the time coordinate lines are the Killing vector's, this corresponds to considering 
static space-times. Our analysis can be extended to the presence of vorticity, that is to stationary space-times.

In AdS rotating thermal states have been studied for free scalars \cite{Thompson:2025jkn} and fermions \cite{Ambrus:2021eod}. 
A similar calculation has been performed for a conformally coupled scalar field in the Closed Einstein Universe \cite{Panerai:2015xlr}, which is not maximally symmetric.

\subsection{The Closed Einstein Universe}

The Closed Einstein Universe (CEU) is a solution of the Einstein field equation with constant proper energy density,
pressure and a positive cosmological constant. Its metric element is given by:
\begin{equation*}
    \di s^2=\di t^2-a^2 \di \chi^2-a^2\sin^2{\chi}\left(\di \theta^2+\sin^2{\theta}\, \di \varphi^2\right)
\end{equation*}
where $t\in\mathbb{R}$, the two angular variables $\chi,\theta \in[0,\pi]$ and $\phi\in[0,2\pi)$
and the Ricci tensor, in the above coordinates, reads:
\begin{equation}\label{RicciCEU}
    R_{\mu\nu}=2\kappa^2\left(g_{\mu\nu}-\delta^{0}_{\mu}\delta^{0}_{\nu}\right)   \; 
\end{equation}
with $\kappa \equiv 1/a$. The curvature scalar is given by $R=6\kappa^2$.

This space-time is globally hyperbolic and conformally flat (so the Weyl tensor vanishes) and can be identified with the product 
$\mathbb{R}\times S^3$. By means of the transformation $\chi=2\arctan{\left(\kappa r/2\right)}$, the metric can be conveniently
recast as:
\begin{equation*}
    \di s^2=\di t^2-\frac{1}{(1+\kappa^2r^2/4)^2}\di r^2-\frac{r^2}{(1+\kappa^2r^2/4)^2}
    (\di \theta^2+\sin^2{\theta}\di \varphi^2) \; ; 
\end{equation*}
in the limit $\kappa\to0$ one recovers the flat space-time line element.

We are going to focus on a state of global equilibrium with vorticity, described by \cite{Panerai:2015xlr}:
\begin{equation}\label{KillingCEU}
    \beta^\mu=\frac{1}{T_0}\left(1,0,0,\kappa\omega_0\right)
\end{equation}
which is globally timelike if $\omega_0<1$ because:
\be\label{esutemp}
    \beta^2 = g_{\mu\nu} \beta^\mu \beta^\nu = \frac{1}{T_0^2} 
    \left( 1-\frac{\kappa^2 r^2 \sin^2 \theta}{(1+\kappa^2r^2/4)^2} \omega_0^2 \right) = \frac{1}{T_0^2} 
    \left( 1- S \omega_0^2 \right)
\ee
where we have introduced $S$:
\begin{equation*}
    S(r,\theta) \equiv \frac{\kappa^2 r^2}{(1+\kappa^2r^2/4)^2}\sin^2\theta
\end{equation*}
for later convenience. The squared four-temperature in \eqref{esutemp} has a minimum for $\kappa r=2, \sin \theta =1$
and it is thus always positive if $\omega_0<1$. The tetrad associated to the Killing vector field 
$\beta$ is given by:
\begin{align}\label{Killingtetrad}
    u^{\mu} &=\left(1-\frac{\omega_0^2\kappa^2 r^2}{(1+\kappa^2r^2/4)^2}\sin^{2}{\theta}\right)^{-1/2}
    \left(1,0,0,\kappa\omega_0\right) \nonumber \\
    A^{\mu}& =-\kappa^2\omega_0^2\left(1-\frac{\omega_0^2 \kappa^2 r^2}{(1+\kappa^2r^2/4)^2}\sin^{2}{\theta}\right)^{-1}
    \left(0,r\frac{(1-\kappa^2r^2/4)}{(1+\kappa^2r^2/4)}\sin^2{\theta},\sin{\theta}\cos{\theta},0\right) \nonumber \\
    \omega^{\mu}& =\kappa\omega_0\left(1-\frac{\omega_0^2\kappa^2 r^2}{(1+\kappa^2r^2/4)^2}\sin^{2}{\theta}\right)^{-1}\left(0,\left(1+\kappa^2r^2/4\right)\cos{\theta},-\sin{\theta}\frac{\left(1-\kappa^2r^2/4\right)}{r},0\right) \nonumber \\
    l^{\mu}&=\kappa^2\omega_0^2\left(1-\frac{\omega_0^2\kappa^2 r^2}{(1+\kappa^2r^2/4)^2}
    \sin^{2}{\theta}\right)^{-5/2}\sin^{2}{\theta}\left(1-\frac{\kappa^2r^2}{(1+\kappa^2r^2/4)^2}\right)
    \left(\frac{\omega_0^2\kappa^2 r^2}{(1+\kappa^2r^2/4)^2}\sin^2{\theta},0,0,\kappa\omega_0\right)
\end{align}
Note that:
$$
  A \cdot \omega = 0
$$
so that $A$ and $\omega$ are linearly independent and the tetrad is non-degenerate.

It is possible to leverage residual symmetries to write these objects in a more compact way by studying the orbits of the Killing vector \cite{Panerai:2015xlr}:
\begin{equation*}
    v=-\frac{(1+\kappa^2r^2/4)}{\kappa}\cos{\theta}\partial_r + \frac{(1-\kappa^2r^2/4)}{\kappa r}\sin{\theta}\partial_\theta
\end{equation*}
For every non vanishing 2D vector field the orbits are given by the contour lines of a scalar function $f(r,\theta)$. Solving:
\begin{equation*}
    -\frac{(1+\kappa^2r^2/4)}{\kappa}\cos{\theta}\partial_r f + \frac{(1-\kappa^2r^2/4)}{\kappa r}\sin{\theta}\partial_\theta f = 0
\end{equation*}
we obtain
\begin{equation*}
    f=\frac{\kappa r}{(1+\kappa^2r^2/4)}\sin\theta 
\end{equation*}
Therefore, given $(r,\theta)$ we can perform a rotation along the orbits of $v$, identified by $(t,r,\theta,\phi) \to (t,\alpha(r,\theta),\pi/2,\phi)$, where $\alpha(r,\theta)$ is such that
\begin{equation}\label{Scapital}
    \frac{\kappa^2 \alpha^2}{(1+\kappa^2\alpha^2/4)^2}=\frac{\kappa^2 r^2}{(1+\kappa^2r^2/4)^2}\sin^2\theta = S(r,\theta)
\end{equation}
%

\subsection{Stress-energy tensor and distillation}

The CEU is globally hyperbolic and the field quantization mirrors the flat space-time case.
No boundary conditions need to be considered and the vacuum of the Hamiltonian operator:
$$
  \widehat H \equiv \int_\Sigma \di \Sigma_\mu \wT^{\mu\nu} \xi_{0\nu}
$$
with $\xi_0^\nu = (1,0,0,0)$ coincides with the vacuum of the operator
$$
  \int_\Sigma \di \Sigma_\mu \wT^{\mu\nu} \beta_{\nu}
$$
with $\beta$ given by the equation \eqref{KillingCEU}. The vacuum-subtracted stress-energy tensor for the state of global 
thermodynamical equilibrium described by \eqref{KillingCEU} can be decomposed onto the tetrad $\{u,A,\omega,l\}$ (see discussion 
in subsection \ref{DistAdsCC}) \cite{Panerai:2015xlr}:
\begin{equation}\label{setesu}
    \langle\hat{T}^{\mu\nu}\rangle=\rho u^{\mu}u^{\nu}-p\Delta^{\mu\nu}+W\omega^{\mu}\omega^{\nu}+\mathcal{A}A^{\mu}A^{\nu}
    +G (u^{\mu}l^{\nu}+u^{\nu}l^{\mu})
\end{equation}
For a real conformally coupled scalar field the scalar thermodynamic functions are given by (we introduce the notation 
$\bar\kappa=\kappa\beta_0$):
\begin{align}\label{esuseries}
    \rho & = \frac{\kappa^4}{12\pi^2(1-S\omega_0^2)}\sum_{n=1}^{\infty}
    \frac{R(n\bar\kappa)}{[1-S-\cosh{(n\bar\kappa)}+
    S\cosh{(n\bar\kappa\omega_0)}]^3}, \nonumber \\
    p &= \frac{\kappa^4}{24\pi^2(1-S\omega_0^2)}
    \sum_{n=1}^{\infty}\frac{P(n\bar\kappa)}{[1-S-\cosh{(n\bar\kappa)}+S\cosh{(n\bar\kappa\omega_0)}]^3}  \nonumber \\
     \mathcal{A} &= \frac{\kappa^2(1-S\omega_0^2)}{6\pi^2(1-S)\omega_0^4} \sum_{n=1}^{\infty}
     \frac{F_A(n\bar\kappa)}{[1-S-\cosh{(n\bar\kappa)}+S\cosh{(n\bar\kappa\omega_0)}]^3} \nonumber \\
     W &= \frac{\kappa^2(1-S\omega_0^2)}{6\pi^2(1-S)\omega_0^2}\sum_{n=1}^{\infty}
     \frac{F_W(n\bar\kappa)}{[1-S-\cosh{(n\bar\kappa)}+S\cosh{(n\bar\kappa\omega_0)}]^3}
     \nonumber \\
     G &=\frac{\kappa^2(1-S\omega_0^2)}{6\pi^2(1-S)\omega_0^3}\sum_{n=1}^{\infty}
     \frac{F_q(n\bar\kappa)}{[1-S-\cosh{(n\bar\kappa)}+S\cosh{(n\bar\kappa\omega_0)}]^3}
\end{align}
where{\footnote{We point out that in the original paper ref.\cite{Panerai:2015xlr} the transverse coefficient $G$ is 
not reported correctly due to a typographical error.}}
\begin{align*}
    R(n\bar\kappa) &= +2(1+S^2\omega_0^2)+(-2+S+S^2\omega_0^2)\cosh{(n\bar\kappa)}+ \nonumber \\
    &+(1+S\omega_0^2-2S^2\omega_0^2)\cosh{(n\bar\kappa\omega_0)}-(1+S+S\omega_0^2+S^2\omega_0^2)
    \cosh{(n\bar\kappa)}\cosh{(n\bar\kappa\omega_0)}+ \nonumber\\
    & -(3+S\omega_0^2)\sinh^2{(n\bar\kappa)}-
    S(1+3S\omega_0^2)\sinh^2{(n\bar\kappa\omega_0)}+12S\omega_0\sinh{(n\bar\kappa)}\sinh{(n\bar\kappa\omega_0)} 
    \nonumber \\ \nonumber \\
   P(n\bar\kappa) &= + 1 + 7 S + S (7 + S) \omega_0^2 - (1 + 3 S \omega_0^2) \cosh\big(2 n \bar\kappa\big) + 
   2 \big(1 + S (-2 + \omega_0^2)\big) \cosh\big(n \, \omega_0 \, \bar\kappa\big) +\nonumber \\
  & - 2 \cosh\big(n \bar\kappa\big) \Big[ S \big(-1 - (-2 + S) \omega_0^2\big) + (1 + S) (1 + S \, \omega_0^2) 
\cosh\big(n \, \omega_0 \, \bar\kappa\big) \Big] +\nonumber \\
  & - S (3 + S \, \omega_0^2) \cosh\big(2 n \, \omega_0 \, \bar\kappa\big) + 24 S \, \omega_0 \, 
  \sinh\big(n \bar\kappa\big) \sinh\big(n \, \omega_0 \, \bar\kappa\big) \nonumber \\ \nonumber \\
  F_A(n\bar\kappa) & = -2(1+\omega_0^2) + \omega_0^2 \cosh\!\left(2 n \bar\kappa\right) + \cosh\!\left(n \bar\kappa\right) 
\left(-1+\omega_0^2 + (1+\omega_0^2)\cosh\!\left(n \omega_0 \bar\kappa\right)\right) + \nonumber \\ 
& - (-1+\omega_0^2)\cosh\!\left(n \omega_0 \bar\kappa\right) + \cosh\!\left(2 n \omega_0 \bar\kappa\right)
- 6 \omega_0 \sinh\!\left(n \bar\kappa\right)\sinh\!\left(n \omega_0 \bar\kappa\right) \nonumber \\ \nonumber \\
 F_W(n\bar\kappa) &= 1 - 3S + (-3 + S)S \omega_0^2 + S \omega_0^2 \cosh\!\left(2 n \bar\kappa\right)
- \left(1 + S(-2 + S \omega_0^2)\right)\cosh\!\left(n \omega_0 \bar\kappa\right) + 
  S \cosh\!\left(2 n \omega_0 \bar\kappa\right) +\nonumber \\
 &+ \cosh\!\left(n \bar\kappa\right) 
 \left(-1 - (-2 + S)S \omega_0^2 + (1 + S^2 \omega_0^2)\cosh\!\left(n \omega_0 \bar\kappa\right)\right)
- 6 S \omega_0 \sinh\!\left(n \bar\kappa\right)\sinh\!\left(n \omega_0 \bar\kappa\right) 
\end{align*}
\begin{align*}
    F_q(n\bar\kappa)& = (1+S)\omega_0-(1-S)\omega_0\cosh{(n\bar\kappa)}+(1-S)\omega_0\cosh{(n\bar\kappa\omega_0)}-(1+S)\omega_0\cosh{(n\bar\kappa)}\cosh{(n\bar\kappa\omega_0)}+
    \nonumber   \\
    &-2\omega_0\sinh^2{(n\bar\kappa)}-2S\omega_0\sinh^2{(n\bar\kappa\omega_0)}+3(1+S\omega_0^2)\sinh{(n\bar\kappa)}
  \sinh{(n\bar\kappa\omega_0)} 
\end{align*}
and $S$ is given by the equation \eqref{Scapital}. Note that the stress-energy tensor \eqref{setesu} is trace-free because 
the conformal anomaly vanishes identically in the CEU \cite{Brown:1977sj}. 

According to the general discussion in subsection \ref{DistAdsCC}, we should distill the stress-energy tensor 
with respect to the derivatives of the four-temperature and the curvature tensors $D\beta$ and $Dg$ but it is more 
convenient to find analytic relations between the parameters appearing in those functions, i.e. $\kappa$ and $\omega_0$ 
and $D\beta,Dg$ without further coordinate dependence. Now,  according to \eqref{RicciCEU}, $\kappa^2$ can be written as 
proportional to the curvature scalar:
$$
    \kappa^2 = \frac{R}{6}
$$
while $\kappa \omega_0$ can be simply written as $T_0 \beta^3$ according to the equation \eqref{KillingCEU};
since we do not require analytic dependence on $\beta$, the variable $\kappa \omega_0$ will not be involved in the distillation 
process. If the scalar functions in \eqref{esuseries} do not have poles for $\kappa^2=0$, we can take advantage of the theorem 
\ref{thm1} and apply distillation only to them because the tensor combinations in \eqref{setesu} are analytic in $\kappa^2=0$.

Before showing the results of the analytic distillation in the general case with $\omega_0 \ne 0$, we first address the 
case with $\omega_0 = 0$. In this limit, the global equilibrium stress-energy tensor with the Killing vector 
$\beta=\beta_0\partial_t$ is recovered:
\begin{equation*}
    \langle\hat{T}_{\mu\nu}\rangle = \rho u_\mu u_\nu - \frac{\rho}{3}\Delta_{\mu\nu}
\end{equation*}
where:
\begin{equation*}
    \rho = \frac{\kappa^4}{16\pi^2}\sum_{n=1}^{\infty}\frac{2+\cosh{(n\kappa\beta_0})}{\sinh^4{(n\kappa\beta_0/2)}}.
\end{equation*}
is independent of space-time coordinates. The above function does not have a pole in $\kappa^2=0$ and its distillation can be 
worked out much the same way as for the AdS case, by using Corollary \ref{coro1}. The results is:
\begin{equation*}
   \dist_{\kappa^2=0} \rho = \frac{\kappa^4}{16\pi^2}\left(\frac{8\pi^2}{15\kappa^4\beta_0^4}-\frac{1}{30}\right)
\end{equation*}
so that:
\begin{equation}\label{distHomogeneousCEU}
    \dist_{\kappa^2=0}  \langle\hat{T}_{\mu\nu}\rangle = \left[\frac{\pi^2}{90\beta_0^4}-\frac{\kappa^4}{1440\pi^2}\right]
    (-g_{\mu\nu}+4u_{\mu}u_{\nu}).
\end{equation}
which agrees with the known high temperature limiting expression \cite{Altaie:1978dx}. Note that, once again, after
distillation the stress-energy tensor no longer vanishes in the limit $T_0 \to 0$ but for the finite value $T_0 = \kappa/2\pi$.

If $\omega_0 \ne 0$, again all the scalar functions in \eqref{series1} do not have poles in $\kappa^2=0$ (this can be seen
during the asymptotic expansion procedure) and we can distill them by using Corollary \ref{coro1}, introducing auxiliary 
variables like in the distillation of the series $E_m$ for the AdS case, see subsection \ref{DistAdsCC}. The expansion has been performed using Mathematica and the final results read:
\begin{align}\label{distCEUcoeff}
   \dist_{\kappa^2=0} \rho = & \frac{\pi^2}{30(1-S\omega_0^2)^2\beta_0^4}+
   \frac{(1-S)\omega_0^2\kappa^2}{36(1-S\omega_0^2)^3\beta_0^2} \nonumber \\
    &-\frac{3S^2\omega_0^8-4S(5S-8)\omega_0^6+10(S^2-5S+1)\omega_0^4+4(8S-5)\omega_0^2+3}{1440\pi^2(1-S\omega_0^2)^4}\kappa^4
 \nonumber \\
   \dist_{\kappa^2=0} p = &\frac{\pi^2}{90(1-S\omega_0^2)^2\beta_0^4} + 
   \frac{(1-S)\omega_0^2\kappa^2}{36(1-S\omega_0^2)^3\beta_0^2} \nonumber \\
    &-\frac{S^2\,\omega_0^8 + 4(16 - 5S)S\,\omega_0^6 + 10\big(1 -11S + S^2\big)\,\omega_0^4 + 4(-5 + 16S)\,\omega_0^2 + 1}
    {1440\pi^2(1-S\omega_0^2)^4}\kappa^4
  \nonumber \\
    \dist_{\kappa^2=0} \mathcal{A} & = \frac{\kappa^2(1-S\omega_0^2)}{6\pi^2(1-S)\omega_0^4}
    \left[\frac{\omega_0^2(1-\omega_0^2)^2}{4(1-S\omega_0^2)^4}\right]
   \nonumber \\
    \dist_{\kappa^2=0} W &= \frac{\kappa^2(1-S\omega_0^2)}{6\pi^2(1-S)\omega_0^2}\left[-\frac{\pi^2(1-S)\omega_0^2}
    {3(1-S\omega_0^2)^2\kappa^2\beta_0^2}-\frac{\omega_0^2(2S^2\omega_0^4-5S\omega_0^4-S^2\omega_0^2+8S\omega_0^2-\omega_0^2-5S-2)}    
    {12(1-S\omega_0^2)^3}\right]
  \nonumber \\
   \dist_{\kappa^2=0} G &= \frac{\kappa^2(1-S\omega_0^2)}{6\pi^2(1-S)\omega_0^3}\left[\frac{\pi^2\omega_0(1-\omega_0^2)}
   {3(1-S\omega_0^2)^2\kappa^2\beta_0^2}+\frac{\omega_0(\omega_0^4-1)}{60(1-S\omega_0^2)^2}\right]
\end{align}
Just like in the AdS case, in the asymptotic expansion of the stress-energy tensor there are contributions from the 
integral term in Corollary \ref{coro1}. In case of $\omega_0=0$ such terms vanish, while in the $\omega_0 \ne 0$ case 
they do not and are discarded by analytic distillation because they are proportional to an odd power of $\kappa$, hence 
non-integer power of $\kappa^2$; they are listed in Appendix \ref{app:covform}. As has been mentioned, there is no need of a 
further distillation in $\kappa \omega_0$; indeed, the poles in $\omega_0$ are canceled by the corresponding powers in the tensor 
combinations of the tetrad $\{u,A,\omega,l\}$ and the stress-energy tensor remains finite for $\omega_0 = 0.$

Like for the AdS and dS case, it can be checked that:
$$
   \dist_{\kappa^2=0} \langle \wT^\mu_\mu \rangle = 0
$$
and
$$
   \nabla_\mu \dist_{\kappa^2=0} \langle \wT^{\mu\nu} \rangle = 0
$$
just like for the original stress-energy tensor in eq. \eqref{setesu} with the scalar functions in \eqref{esuseries}.

\section{Universal form of the distillate of the stress-energy tensor}
\label{sec:uni}

So far, we have obtained the expressions of the analytic distillates of the stress-energy tensor of the conformally and minimally
coupled scalar free field in several space-times. Now, we would like to compare them with each other by re-expressing them as a 
function of tensor, vector and scalar fields obtained from the derivatives of the Killing vector field as well as from the 
curvature tensors in a fully covariant fashion. The goal is to verify the conjectured universality, meaning that the resulting 
expression is the same independently of the space-time. 

Indeed, once the expression of the stress-energy tensor in a specific space-time is found, the corresponding covariant expression 
is not unique. This especially applies to highly symmetric space-times. For instance, as has been mentioned in Section \ref{sec:ads}, 
for the four-temperature \eqref{KillingAdS} in AdS the following degeneracy appears:
\begin{equation}\label{deg}
    R_{\mu\nu} = \frac{1}{4}Rg_{\mu\nu} = \left(R_{\rho\sigma}u^\rho u^\sigma\right) g_{\mu\nu}
\end{equation}
meaning that $\kappa^2$ can be replaced with either $-R/12$ or $-(1/3) R_{\mu\nu}u^\mu u^\nu$ as the argument of all scalar functions.

In order to check the conjectured existence of a universal analytic expression of the stress-energy tensor we then employ the 
following strategy: at each order in the derivatives (of the four-temperature and the metric), a general expression of the 
stress-energy tensor is written down in terms of a linear combination of all the possible tensors and scalar functions with unknown 
coefficients. The coefficients are pinned down by comparing the general expression with our distilled results. Through this method, 
at each order, a system of equations for the unknown coefficients is obtained and their universality can be verified. Note that not 
all coefficients are actually independent, for the stress-energy tensor is required to satisfy constraints such as the conservation 
equation and is invariant under some symmetry transformations.

\subsection{Conformally coupled field}

At the zeroth order in derivatives the most general form of the stress-energy tensor reads:
\begin{equation}\label{zerorder}
    \langle\hat{T}_{\mu\nu}\rangle^{(0)} = \rho^{(0)} u_\mu u_\nu - p^{(0)}\Delta_{\mu\nu}
\end{equation}
where $\rho^{(0)}=a_1T^4$ and $p^{(0)}=a_2T^4$. It is straightforward to check that in all the analytic distillates found
in the examined space-times we have:
\begin{equation*}
 a_1=\frac{a_2}{3}=\frac{\pi^2}{30}
\end{equation*}

No first order term appears in the expressions of the stress-energy tensor. At the second order in derivatives one can write the 
following general form, which is consistent with transformations under parity:
\begin{align}\label{generalSecondOrder}
    \langle\hat{T}_{\mu\nu}\rangle^{(2)} =&\rho^{(2)} u_{\mu}u_{\nu}-p^{(2)}\Delta_{\mu\nu}+\mathcal{A}^{(0)}A_{\mu}A_{\nu}+W^{(0)}\omega_{\mu}\omega_{\nu}+p_R^{(0)}R_{\mu\nu} +q^{(0)}\left(u_{\mu}l_{\nu}+u_{\nu}l_{\mu}\right)\nonumber \\
    &+p_{Ru}^{(0)}\left(u_{\mu} R_{\nu}^{\ \ \alpha} u_\alpha + u_\nu R_{\mu}^{\ \ \alpha}u_\alpha\right) + p_{C}^{(0)}C_{\mu\rho\nu\sigma}u^{\rho}u^{\sigma},
\end{align}
where:
\begin{align*}
    \rho^{(2)}&=b_1T^2A^2+b_2T^2\omega^2+b_3T^2R+b_4T^2R^{\mu\nu}u_\mu u_\nu   \\
    p^{(2)}&=c_1T^2A^2+c_2T^2\omega^2+c_3T^2R+c_4T^2R^{\mu\nu}u_\mu u_\nu
\end{align*}
$C_{\mu\rho\nu\sigma}$ being the Weyl tensor, and:
\begin{equation*}
    \mathcal{A}^{(0)}=d_1T^2, \ \ W^{(0)}=d_2T^2, \ \ q^{(0)}=d_3T^2, \ \ p_R^{(0)}=d_4T^2, \ \ p_{Ru}^{(0)}=d_5T^2, \ \ p_{C}^{(0)}=d_6T^2
\end{equation*}
All of our results have been obtained in conformally flat space-times, so $C_{\mu\rho\nu\sigma}=0$ and the coefficient $d_6$
cannot be determined. 

The distillate of the stress-energy tensor in AdS at global equilibrium \eqref{distadsset} with the Killing vector \eqref{KillingAdS} vanishes at the second order in derivatives. Thus, the \eqref{generalSecondOrder} implies:
\begin{align*}
    \beta^2\langle\hat{T}_{\mu\nu}\rangle^{(2)} =& \left(b_1A^2+b_3R+b_4R^{\rho\sigma}u_\rho u_\sigma\right)u_{\mu}u_{\nu} - \left(c_1A^2+c_3R+c_4R^{\rho\sigma}u_\rho u_\sigma\right)\Delta_{\mu\nu} +\\
    &+ d_1A_\mu A_\nu + d_4R_{\mu\nu} + d_5\left(u_{\mu} R_{\nu}^{\ \ \alpha} u_\alpha + u_\nu R_{\mu}^{\ \ \alpha}u_\alpha\right) = \\ 
    =&-\kappa^2\left(\frac{r^2\kappa^2}{1+r^2\kappa^2}(b_1+c_1)+3(4b_3+4c_3+b_4+c_4+2d_5)\right)u_\mu u_\nu + \\
    &-\kappa^2\left(\frac{r^2\kappa^2}{1+r^2\kappa^2}c_1+3(4c_3+c_4+d_4)\right)g_{\mu\nu}+d_1A_{\mu}A_{\nu} = 0
\end{align*}
whence:
\begin{equation*}
b_1=c_1=d_1=0, \ \ \ 
4b_3+b_4+2d_5 = 0, \ \ \
4c_3+c_4+d_4 = 0. 
\end{equation*}
The very same solution is obtained from the distillate in dS space-time \eqref{dsset}. Following the same strategy (see
Appendix \ref{app:covform} for details) for the distillate in the CEU (see eqs. \eqref{setesu} and \eqref{distCEUcoeff}) 
for both the $\omega_0=0$ and $\omega_0 \ne 0$ case, one obtains an over-determined system of equations which,
nevertheless, admits a solution. The final, universal, expression at the second order reads:
\begin{align}\label{2ndorder}
    \langle\widehat{T}_{\mu\nu}\rangle^{(2)} =&\left(-\frac{1}{36}T^2\omega^2+\frac{1}{18}T^2R^{\rho\sigma}u_\rho 
    u_\sigma\right)u_{\mu}u_{\nu}+\frac{1}{36}T^2\omega^2\Delta_{\mu\nu}-\frac{1}{18}T^2\omega_{\mu}\omega_{\nu}+ \nonumber \\
    &+\frac{1}{18}T^2(u_{\mu}l_{\nu}+u_{\nu}l_{\mu})-\frac{1}{36}T^2\left(u_{\mu} R_{\nu}^{\ \ \alpha} u_\alpha+u_{\nu} 
    R_{\mu}^{\ \ \alpha} u_\alpha\right)
\end{align}
Remarkably, this expression exactly reproduces the flat space-time one \eqref{minkset} with the scalars \eqref{minkfunc}
by setting the curvature tensors equal to zero, in full agreement with the universality conjecture. Note that the equation 
\eqref{2ndorder} matches the high temperature limit obtained in \cite{Fursaev:2001yu} and a perturbative result obtained 
through Kubo formulae in \cite{Kovtun:2018dvd}. Thanks to the calculation in CEU, we can then conclude that the second order 
term vanishes in both dS and AdS with vanishing vorticity $\omega_\mu=0$, because the Ricci tensor is proportional to the 
metric tensor, leading to the cancellation of the terms:
$$
 \frac{1}{18}T^2R^{\rho\sigma}u_\rho u_\sigma u_{\mu}u_{\nu}-\frac{1}{36}T^2\left(u_{\mu} 
 R_{\nu}^{\ \ \alpha} u_\alpha+u_{\nu} R_{\mu}^{\ \ \alpha} u_\alpha\right) = 0
$$

The next, and highest, non-vanishing order in powers of derivatives is the fourth. In principle, the same method as for
the second order can be used, but the number of possible terms is significantly larger. In practice, the expressions
found in the examined space-times are not sufficient to pin down all the coefficients at this order. Notwithstanding, it 
is possible to determine the coefficients surviving in the flat space-time case at the fourth order by using dS, AdS and CEU 
results. The final result reads:
\begin{align*}
    \langle\widehat{T}_{\mu\nu}\rangle^{(4)} =& -\frac{1}{1440\pi^2}\left(3A^4+10\omega^4+32A^2\omega^2\right)u_{\mu}u_{\nu} +\frac{1}{1440\pi^2} \left(A^4+10\omega^4+A^2\omega^2\right)\Delta_{\mu\nu} + \\ \nonumber
    &-\frac{1}{288\pi^2}\left(7A^2+4\omega^2\right)\omega_\mu\omega_\nu +\frac{1}{480\pi^2} A^2A_{\mu}A_{\nu} -\frac{1}{360\pi^2}\left(A^2+\omega^2\right)(u_\mu l_\nu+ u_\nu l_\mu) + \\ \nonumber
    &+\text{curvature corrections}
\end{align*}
which is in full agreement with the equation \eqref{minkset} and \eqref{minkfunc} in the Minkowski space-time, again
bearing out the universality conjecture.

\subsection{Minimally coupled field}

For the minimally coupled scalar field the lack of a  distillate in the presence of rotation prevents the derivation of a universal expression for the stress-energy tensor up to second order in derivatives, as we have done for conformal coupling. However, we can 
still use the distillate in AdS to constrain the coefficients of the general expression following the same procedure described in 
the previous subsection and prove the universality of all the corrections depending on $A^2$ and $A^4$.

At zero order in derivatives, the minimally coupled scalar field mirrors the behavior of the conformal coupling case, as given by 
\eqref{zerorder}. At second order in derivatives the most general form of the stress-energy tensor without vorticity is given by \eqref{generalSecondOrder} with $\omega_\mu=0$, that is:
\begin{equation}\label{2AdSMinimal}
    \langle\hat{T}_{\mu\nu}\rangle^{(2)} =\rho^{(2)} u_{\mu}u_{\nu}-p^{(2)}\Delta_{\mu\nu}+\mathcal{A}^{(0)}A_{\mu}A_{\nu}+p_R^{(0)}R_{\mu\nu} +p_{Ru}^{(0)}\left(u_{\mu} R_{\nu}^{\ \ \alpha} u_\alpha + u_\nu R_{\mu}^{\ \ \alpha}u_\alpha\right),
\end{equation}
with:
\begin{align*}
    \rho^{(2)}&=b_1T^2A^2+b_3T^2R+b_4T^2R^{\mu\nu}u_\mu u_\nu   \\
    p^{(2)}&=c_1T^2A^2+c_3T^2R+c_4T^2R^{\mu\nu}u_\mu u_\nu
\end{align*}
\begin{equation*}
    \mathcal{A}^{(0)}=d_1T^2, \ \  p_R^{(0)}=d_4T^2, \ \ p_{Ru}^{(0)}=d_5T^2.
\end{equation*}
We can use the previously described matching procedure to constrain these coefficients. The same thing can be done at 
order four in derivatives, more details can be found in Appendix \ref{app:covform}. Eventually, all coefficients surviving 
in the flat space-time limit can be identified and some constraints can be obtained on the others; they again agree with the
equation \eqref{minkfunc2}. 

Among the several covariant forms of the stress-energy tensor in AdS, we pick the one in which the Ricci tensor does 
not appear, which is given by:
\begin{equation}\label{adsminimset}
    \langle\widehat{T}_{\mu\nu}\rangle = \rho u_{\mu}u_{\nu} - p\Delta_{\mu\nu} + \mathcal{A}A_{\mu}A_{\nu}
\end{equation}
with:
\begin{align}\label{adsminimscalar}
    \rho &= \frac{\pi^2}{30\beta^4}+\frac{1}{12\beta^2}\left(\frac{5}{2}\left(\frac{R}{12}\right)-A^2\right)
    +\frac{1}{480\pi^2}\left(-71\left(\frac{R}{12}\right)^2+37\left(\frac{R}{12}\right)A^2-11A^4\right)
   \nonumber \\
    p &= \frac{\pi^2}{90\beta^4} + \frac{1}{18\beta^2}\left(A^2-\frac{1}{4}\left(\frac{R}{12}\right)\right) 
    + \frac{1}{1440\pi^2}\left(139\left(\frac{R}{12}\right)^2-23\left(\frac{R}{12}\right)A^2+19A^4\right)
   \nonumber \\
    \mathcal{A} &=\frac{1}{12\beta^2} - \frac{1}{48\pi^2}\left(\left(\frac{R}{12}\right)-A^2\right)
\end{align}
This form can be further converted to highlight the role of the dS Unruh temperature, see Section \ref{sec:unruh}.

\section{Discussion and the Unruh effect}
\label{sec:unruh}

One of the most remarkable features of the universal, exact solution of the distillate of $\langle \wT^{\mu\nu} \rangle$ is that it 
is a truncated expansion in the curvature tensors and the derivatives of the Killing vector, like in the flat space-time case. 
The expansion in $A^\mu,\omega^\mu$ and curvature tensors actually stops at the fourth order, as assumed in ref. 
\cite{Khakimov:2023emy}. Terms which, in principle, are analytic in the derivatives and give rise to an infinite series in 
the curvature such as, for instance:
$$
   \frac{T^4A^2}{aT^2+bR}
$$
with $a$ and $b$ arbitrary real constants, do not appear in the general expression. 

As has been stressed, the coefficients of all the terms appearing in the thermodynamic scalar functions in curved space-time 
are the same as in Minkowski space-time in all examined cases, bearing out the universality conjecture. In fact, a difference 
between the stress-energy tensor in Minkowski space-time in eq. \eqref{minkset} and those in AdS, dS, CEU should be spelled out: 
in Minkowski, the scalars \eqref{minkfunc}, \eqref{minkfunc2} have been obtained by subtracting the expectation value in the 
so-called Minkowski vacuum $\ket{0_M} $(see normal ordering in \eqref{minkset}) pertaining to constant Killing field $\beta^\mu=b^\mu$
and not to the Killing vector with acceleration or vorticity $\beta_\mu = b_\mu + \varpi_{\mu\nu} x^\nu$. 
Conversely, in curved space-time, the stress-energy tensor has been calculated 
by subtracting the expectation value in the {\em same} vacuum pertaining to the Killing vector under consideration, possibly endowed
with an acceleration. The full correspondence between flat and curved cases likely resides in the nature of the examined vacua: 
in AdS and CEU case, the vacua correspond to Killing vectors \eqref{KillingAdS}, \eqref{KillingCEU} which are globally time-like 
and maximally invariant under the symmetry transformation of the space-time \cite{Avis:1977yn}, just like the Minkowski vacuum. In dS, the 
vacuum associated to \eqref{KillingdS} is not maximally symmetric, but it corresponds to the maximal region (the so-called static 
patch) covered by a time-like Killing vector. Because of this similarity, one should then expect similar terms in
the expansion.

Another remarkable feature is that the distilled $\langle \wT^{\mu\nu} \rangle$ does not vanish for $T_0=0$, and consequently
for $T=0$. For instance, in the Minkowski space-time, the functions \eqref{minkfunc} and \eqref{minkfunc2} do not vanish at 
$T=0$; the same happens for the AdS and dS expressions in \eqref{distadsset} and \eqref{dsdistset} as well as in the CEU 
space-time \eqref{distCEUcoeff}. In fact, for $\omega^\mu=0$, corresponding to $\omega_0=0$ in the CEU case, all the distilled 
expressions of the stress-energy tensor in the conformal coupling case in eq. \eqref{minkset},\eqref{distadsset},\eqref{dsdistset} 
and \eqref{distHomogeneousCEU} vanish at some finite value of $T_0$ or $T$. In Minkowski space-time, according to \eqref{minkset}, 
this happens at the comoving Unruh \cite{Unruh:1976db} temperature, that is for:
\be\label{unruh}
  T_U = \frac{1}{2\pi} \sqrt{-A^2}
\ee
as observed in ref. \cite{Becattini:2017ljh}. The same property has been observed also in the massive case \cite{Prokhorov:2019yft}.
We point out that in the massless case the distilled $\langle\widehat{T}_{\mu\nu}\rangle$ \eqref{minkset} matches the exact result 
in Rindler space-time, while in the massive case this property holds up to the fourth order in derivatives.

The existence of a well-defined Unruh temperature has been rigorously established for observers in AdS and dS space-times 
\cite{Deser:1997ri,Narnhofer:1996zk,Bros:2001tk}:
\begin{equation*}
    T_U^2=\frac{1}{4\pi^2}\left(\frac{R}{12}-A^2\right)
\end{equation*}
In sharp contrast to the Minkowski case, where any non-vanishing acceleration yields a thermal response, in AdS there is a 
critical acceleration $A^2_{\text{crit}}=-\kappa^2$. Observers accelerating with $A^2>A^2_{\text{crit}}$ do not experience 
any thermal bath. This behavior can be easily understood in terms of bifurcate Killing horizons \cite{Jacobson:1997ux}. 
Indeed the Killing vector \eqref{KillingAdS} is globally timelike and therefore no horizon is present. Observers following 
its orbits have accelerations always above the critical acceleration. This means that $T^2_U<0$ so it is not possible to 
interpret $T^2_U$ as a physical temperature. However, starting from the embedding \eqref{adsembedding}, it is possible to 
consider the following embedding coordinates \cite{Deser:1998xb}:
\begin{equation*}
      \begin{aligned}
    X_{0}&=\sqrt{\xi^2-\kappa^{-2}}\sinh{\kappa\eta}, \  \
    X_{1}=\sqrt{\xi^2-\kappa^{-2}}\cosh{\kappa\eta}, \ \\
    X_{2}&=\xi\sinh{\psi}\cos{\delta}, \ \
    X_{3}=\xi\sinh{\psi}\sin{\delta}, \ \
    X_{4}=\xi\cosh{\psi},
    \end{aligned}
\end{equation*}
where $\eta\in(-\infty,\infty)$, $\psi\in(-\infty,\infty)$, $\delta\in(-\pi,\pi)$ and $\xi>\kappa^{-1}$. The metric is then:
\begin{equation}
    ds^2=(\kappa^2\xi^2-1)d\eta^2-(\kappa^2\xi^2-1)^{-1}d\xi^2-\xi^2\left(d\psi^2+\sinh^2{(\psi)}d\delta^2\right).
    \label{adS Unruh}
\end{equation}
The metric does not depend on $\eta$, so the vector field $\partial_\eta$ is a Killing vector. It is time-like in the region 
$\xi>\kappa^{-1}$ and becomes lightlike $\xi=\kappa^{-1}$, which therefore defines a Killing horizon for $\partial_\eta$. 
Observers following the orbits of  $\partial_\eta$ are vorticity-free and experience constant acceleration below the critical value:
\begin{equation*}
    A^2=-\frac{\kappa^2\xi^2}{(\xi^2-\kappa^{-2})} < A_{\text{crit}}^2
\end{equation*}
The presence of the Killing horizon implies that these observers experience the Unruh effect. In dS the condition $T^2=T_U^2$ 
reproduces the well-known Gibbons-Hawking temperature $T_0=\kappa/2\pi$ , which arises from the cosmological horizon in 
the static patch \cite{Gibbons:1977mu}.

It is natural to ask whether our distilled $\langle\widehat{T}_{\mu\nu}\rangle$ in AdS/dS exhibits a similar behavior to 
the flat space-time result and whether our results are sensitive to the existence of a critical acceleration. In AdS we 
have derived the distilled stress–energy tensor in the global equilibrium state generated by the globally timelike Killing 
vector. However, under the conjectured universality of the distilled expression, it is also possible to investigate the 
behavior associated with $\partial_\eta$. Indeed this Killing vector shares the relevant kinematical properties of the 
globally timelike one, such as non-vanishing acceleration and vanishing vorticity.

In AdS/dS for conformal coupling we can express the analytic part of the stress-energy tensor \eqref{distadsset} as:
\begin{equation} \label{distads}
    \langle \widehat{T}_{\mu\nu}\rangle=\Biggl[\frac{\pi^2}{90}T^4-\frac{1}{1440\pi^{2}}
    \Biggl(\frac{R}{12}-A^2\Biggl)^{2}\Biggl](-g_{\mu\nu}+4u_{\mu}u_{\nu}).
\end{equation}
This equation matches the result of \cite{Khakimov:2023emy} in which a gradient expansion for 
the stress-energy tensor of a conformally coupled scalar field in AdS/dS was obtained by imposing conservation at all orders, 
the correct flat space-time limit, the correct trace anomaly and by making the hypothesis that the gradient expansion 
stops at order four. Here we reproduce the same result directly from quantum field theory in curved space-time, without 
invoking any additional assumptions. The distilled stress-energy tensor in AdS \eqref{distads} vanishes at the Unruh 
temperature $T_U^2$ for both the globally timelike Killing vector and for $\partial_\eta$. In the former case the absence 
of a $T^2$ term ensures the cancellation, even though $T^2_U$ cannot be identified with the Unruh temperature. In the 
latter case, as well as in dS, the vacuum subtracted distilled stress-energy tensor vanishes exactly at the physical 
Unruh temperature, as in flat space-time.

In the minimal coupling case in AdS, the scalar coefficients of the distilled stress-energy tensor can be rewritten as follows:
\begin{align}\label{mincoeff}
    \rho & = \frac{\pi^2}{30}\left(T^4+10T^2T_U^2-11T_U^4\right) + 
    \frac{1}{96}R\left(T^2-T_U^2\right)-\frac{3}{32\pi^2}\left(\frac{R}{12}\right)^2  \nonumber \\
    p & = \frac{\pi^2}{90} \left(T^4-20T^2T_U^2+19T_U^4\right) + 
    \frac{1}{288}R\left(T^2-T_U^2\right) +\frac{3}{32\pi^2}\left(\frac{R}{12}\right)^2  \nonumber \\
    \mathcal{A} & = \frac{1}{12}\left(T^2-T_U^2\right).
\end{align}
For the globally timelike Killing vector it is not possible to identify $T^2=T_U^2$ and there is no value of $T_0$ for 
which the stress-energy tensor vanishes. By contrast, for the state of global equilibrium described by $\partial_\eta$ 
equation \eqref{mincoeff} at $T^2=T_U^2$ yields:
$$
    \langle\hat{T}_{\mu\nu}\rangle (T=T_U) = -\frac{3}{32\pi^2}\left(\frac{R}{12}\right)^2g_{\mu\nu}.
$$
which is finite and remarkably {\em independent} of the Killing vector, suggesting a purely geometrical origin for this remaining term. 

Altogether, the existence of a temperature $T$ at which the stress–energy tensor either vanishes or reduces to a 
Killing-vector–independent contribution does not, by itself, imply the presence of an Unruh effect. However, whenever 
an Unruh effect is associated with the Killing vector under consideration, the distilled stress-energy tensor is found to 
satisfy this property.

\section{Summary and conclusions}
\label{sec:conclu}

In summary we have considered exact solutions of the stress-energy tensor of a free real massless scalar field at global thermodynamic 
equilibrium in different conformally flat space-times and for different equilibrium configurations and we have determined their
analytic part in the gradients of the Killing vector and of the metric tensor by using the method of analytic distillation. 
This method, previously employed successfully in flat space-time, enables to extract the analytic part through an asymptotic 
expansion which, for the quantum field at stake, reduces to a finite sum of terms. After recasting the distilled stress-energy 
tensor in a covariant form, we have found that it fulfills the conservation equation and that the coefficients of its expansion coincide 
in all the examined cases. 

This result bears out a universality conjecture, discussed in Section \ref{sec:conj}, according to which the analytic part in
the gradients of the derivatives of the Killing vector characterizing global equilibrium and the curvature tensors of any 
local operator of any given quantum field theory is independent of the underlying space-time. Otherwise stated, at the analytic thermodynamic response to curvature, acceleration and vorticity appears to be universal. 
On the other hand, the effects of global properties of the space-time and of the boundary conditions appear in non-analytic terms. 
Non-analytic contributions appear both in AdS, where a timelike boundary at infinity is present, and in the Closed Einstein Universe 
in the presence of vorticity. These contributions may be non-negligible with respect to higher orders of the distilled part. 

In the case of conformal coupling we have obtained the most general form of the stress-energy tensor at global equilibrium up to 
the second order in derivatives, which matches the results of previous perturbative calculations. Finally, the analytic 
part of the stress-energy tensor displays a natural connection to the Unruh effect, which is known to be present in AdS, as 
it vanishes for conformal coupling at the Unruh temperature, exactly as in flat space-time, and becomes independent of the 
Killing vector for minimal coupling.

The observed universality may be verified in other cases. If confirmed, it makes it possible to determine the
analytic part of the stress-energy tensor of specific quantum field theories in a curved space-time knowing the solution
elsewhere.

\section*{Acknowledgements}

We are very grateful to V. Ambrus, M. Chernodub, E. Grossi, J. Rausch and D. Wagner for interesting discussions.

%

\appendix

\section{Asymptotic expansion of the boundary contribution}
\label{app:boundaryads}

Taking into account $E_m(y)= C_m(y) + \ii S_m(y)$ and the equation \eqref{emfunc}, it is possible to find the
asymptotic expansion of $C_1$ and $S_2$ by using the techniques based on the Mellin transform:
\begin{equation*}
    C_{1}\left(\beta_{0}\kappa,r\right)\sim-\lim_{y\rightarrow\kappa}\text{Re} \left(\frac{1}{\kappa\beta_{0}}\sum_{t=0}^{\infty}\frac{B_{t}\left(ic\right)}{t!}\eta(-t)(\beta_{0}\kappa)^{t}\right)
\end{equation*}
\begin{equation*}
    S_{2}\left(\beta_{0}\kappa,r\right)\sim-\lim_{y\rightarrow\kappa}\text{Im} \left(\left(\frac{1}{\kappa\beta_{0}}\right)^{2}\sum_{t=1}^{\infty}\frac{B_{t-1}\left(ic\right)}{(t-1)!}\eta(-t)(\beta_{0}\kappa)^{t}\right)
\end{equation*}
Let us focus on $C_1$. The Dirichlet Eta function is zero on negative even integers. The real part of odd 
index Bernoulli polynomials with imaginary argument is given by:
\begin{equation*}
    \text{Re} \left(B_{2p+1}\left(ic\right)\right)=(-1)^{p+1}\frac{1}{2}(2p+1)c^{2p}.
\end{equation*}
Therefore, the asymptotic series for $C_1$ can be written as:
\begin{equation*}
    C_{1}\left(\beta_{0}\kappa,r\right)\sim-\frac{1}{2\beta_{0}\kappa}+\frac{1}{8}-\sum_{l=1}^{\infty}\frac{(-1)^{l+1}(l+\frac{1}{2})}{(2l+1)!}\eta(-(2l+1))(2\arctan(r\kappa))^{2l}
\end{equation*}
Writing the Dirichlet Eta function in terms of the Bernoulli numbers using equation \eqref{etaBern} we can write the series as:
\begin{align*}
    &\sum_{l=1}^{\infty}\frac{(-1)^{l+1}(l+\frac{1}{2})}{(2l+1)!}\eta(-(2l+1))(2y)^{2l} = \sum_{l=1}^{\infty}\frac{(-1)^{l+1}(2l+1)}{(2l+2)!}2^{2l-1}(2^{2l+2}-1)B_{2l+2}y^{2l}=\\
    &= -\frac{1}{8}\sum_{l=1}^{\infty}\frac{(-1)^{l}(2l+1)}{(2l+2)!}2^{2l+2}(2^{2l+2}-1)B_{2l+2}y^{2l}=-\frac{1}{8}\tan^{2}(y)
\end{align*}
where $y=\arctan{r\kappa}$ and we have used the Taylor series for $\tan^2(y)$. The asymptotic expansion for $C_{1}\left(\beta_{0}\kappa,r\right)$ can finally be written as:

\begin{equation*}
C_{1}\left(\beta_{0}\kappa,r\right)\sim-\frac{1}{2\beta_{0}\kappa}+\frac{1}{8}+\frac{1}{8}(r\kappa)^2.
\end{equation*}
The calculation for the asymptotic series of $S_2$ is analogous to the one for $S_0$ and $C_1$ and it's based on \eqref{etaBern}.

\section{Integral terms in the Closed Einstein Universe}
\label{app:CEU}

The integral term in the CEU for the homogeneous equilibrium vanishes exactly
\begin{equation*}
    I=\int_0^{\infty}\left(\frac{2+\cosh{x}}{(\sinh{(x/2)})^4}-\frac{48}{x^4}\right)dx = 0
\end{equation*}
The integral terms that appear in the presence of vorticity in the CEU are given by:
\begin{equation*}
    I_\rho = \int_0^{\infty}\left(\frac{R(x)}{[1-S-\cosh{(x)}+S\cosh{(x\omega_0)}]^3}-\frac{36}{(1-S\omega_0^2)x^4}-\frac{2\omega_0^2(1-S)}{(1-S\omega_0^2)^2x^2}\right)dx
\end{equation*}
\begin{equation*}
    I_p = \int_0^\infty \left(\frac{P(x)}{[1-S-\cosh{(x)}+S\cosh{(x\omega_0)}]^3} -\frac{24}{(1-S\omega_0^2)x^4}-\frac{
    4\omega_0^2(1-S)}{(1-S\omega_0^2)^2x^2}\right)dx
\end{equation*}
\begin{equation*}
    I_\mathcal{A} = \int_0^{\infty}\left(\frac{F_A(x)}{[1-S-\cosh{(x)}+S\cosh{(x\omega_0)}]^3}
    \right)dx
\end{equation*}
\begin{equation*}
    I_W = \int_0^{\infty}\left(\frac{F_W(x)}{[1-S-\cosh{(x)}+S\cosh{(x\omega_0)}]^3}+\frac{2(1-S)\omega_0^2}{(1-S\omega_0^2)^2x^2}\right)dx
\end{equation*}
\begin{equation*}
    I_q = \int_0^{\infty}\left(\frac{F_q(x)}{[1-S-\cosh{(x)}+S\cosh{(x\omega_0)}]^2}-\frac{2\omega_0(1-\omega_0^2)}{(1-S\omega_0^2)^2x^2}\right)dx
\end{equation*}

\section{Covariant form of the stress-energy tensor}
\label{app:covform}

The general form of the stress-energy tensor at second order in derivatives in the non-rotating ($\omega_0=0$) CEU can be
written as:
\begin{align*}
    \beta^2\langle\widehat{T}_{\mu\nu}\rangle^{(2)} =& \ b_3R\ u_{\mu}u_{\nu} - c_3R\ \Delta_{\mu\nu} + d_4R_{\mu\nu} = 6\kappa^2u_\mu u_\nu+2\kappa^2(d_4-3c_3)\Delta_{\mu\nu}
\end{align*}
because for the equilibrium without rotation $u=\partial_t$ hence the Ricci tensor in eq. \eqref{RicciCEU} can be written as
$R_{\mu\nu}=2\kappa^2\Delta_{\mu\nu}$. At this order the distilled expression \eqref{distHomogeneousCEU} vanishes, so we obtain 
the following constraints:
\begin{equation*}
    b_3=0, \qquad \qquad d_4-3c_3=0.
\end{equation*}
In the presence of vorticity, the general form of the stress-energy tensor in the CEU reads:
\begin{align}\label{gen2ord}
    \langle\widehat{T}_{\mu\nu}\rangle^{(2)} =&\rho^{(2)} u_{\mu}u_{\nu}-p^{(2)}\Delta_{\mu\nu}+\mathcal{A}^{(0)}A_{\mu}A_{\nu}+W^{(0)}\omega_{\mu}\omega_{\nu}+p_R^{(0)}R_{\mu\nu}+q^{(0)}
    \left(u_{\mu}l_{\nu}+u_{\nu}l_{\mu}\right)+ \nonumber \\
    &+p_{Ru}^{(0)}\left(u_{\mu} R_{\nu}^{\ \ \alpha} u_\alpha + u_\nu R_{\mu}^{\ \ \alpha}u_\alpha\right) 
\end{align}
where the scalar functions are defined in eq. \eqref{generalSecondOrder}. The above expression is to be matched to the
equation \eqref{setesu} with the scalar functions in \eqref{distCEUcoeff} up to second order in derivatives, so:
\begin{align}\label{ceu2order}
    \dist_{\kappa^2=0}\langle\widehat{T}_{\mu\nu}\rangle^{(2)}_{CEU} =& \left(\frac{(1-S)\omega_0^2\kappa^2}{36(1-S\omega_0^2)^3\beta_0^2}\right)u_{\mu}u_\nu -\left(\frac{(1-S)\omega_0^2\kappa^2}{36(1-S\omega_0^2)^3\beta_0^2}\right)\Delta_{\mu\nu}+\left(-\frac{1}{18(1-S\omega_0^2)\beta_0^2}\right)\omega_\mu\omega_\nu + 
    \nonumber \\
    &+\left(\frac{(1-\omega_0^2)}{18\omega_0^2(1-S\omega_0^2)\beta_0^2}\right)(u_\mu l_\nu + u_\nu l_\mu)
\end{align}
In order to obtain the coefficients, it is convenient to consider the projections of both \eqref{gen2ord} and
\eqref{ceu2order} and equate them. Notably, from \eqref{gen2ord}:
\begin{align*}
    u^{\mu}u^{\nu}\langle\widehat{T}_{\mu\nu}\rangle^{(2)} &= \rho^{(2)} + (p_R^{(0)}+ 2p_{Ru}^{(0)})R^{\mu\nu}u_\mu u_\nu 
    = T^2\left(b_1A^2+b_2\omega^2+b_3R+(b_4+d_4+2d_5)R_{\mu\nu}u^{\mu}u^{\nu}\right) \\
    \hat{A}^\mu\hat{A}^\nu\langle\widehat{T}_{\mu\nu}\rangle^{(2)} & = p^{(2)}-
    \mathcal{A}^{(0)}A^2+p_R^{(0)}R_{\mu\nu}\hat{A}^\mu\hat{A}^\nu = T^2\left((c_1-d_1)A^2+c_2\omega^2+
    c_3R+c_4R_{\mu\nu}u^\mu u^\nu+d_4R_{\mu\nu}\hat{A}^{\mu}\hat{A}^{\nu}\right) \\
    \hat{\omega}^\mu\hat{\omega}^\nu\langle\widehat{T}_{\mu\nu}\rangle^{(2)} &= p^{(2)}-W^{(0)}\omega^2+
    p_R^{(0)}R_{\mu\nu}\hat{\omega}^\mu\hat{\omega}^\nu = T^2\left(c_1A^2+(c_2-d_2)\omega^2+c_3R+
    c_4R_{\mu\nu}u^{\mu}u^{\nu}+d_4R_{\mu\nu}\hat{\omega}^\mu\hat{\omega}^{\nu}\right) \\
    u^\mu \hat{l}^\nu\langle\widehat{T}_{\mu\nu}\rangle^{(2)} &= (p_R^{(0)}+p^{(0)}_{Ru})R_{\mu\nu}u^\mu \hat{l}^\nu 
    - q^{(0)}\sqrt{-l^2} = T^2\left((d_4+d_5)R_{\mu\nu}u^\mu \hat{l}^\nu-d_3\sqrt{-l^2}\right) 
\end{align*}
For instance, the right hand side of the first equation above is to be matched with the corresponding projection of 
the \eqref{ceu2order}:
\begin{equation*}
    u^\mu u^\nu\langle\widehat{T}_{\mu\nu}\rangle_{CEU}^{(2)} = \frac{(1-S)\omega_0^2\kappa^2}{36(1-S\omega_0^2)^3\beta_0^2}
\end{equation*}
and the resulting equation reads, taking into account the eq. \eqref{KillingCEU} and $T^2= 1/\beta^2$
\begin{equation*}
    6b_3-b_2\omega_0^2+(b_2-12b_3-2(b_4+d_4+2d_5))S\omega_0^2+(6b_3-b_1)S\omega_0^4+(2(b_4+d_4+2d_5)+b_1)S^2\omega_0^4
    =\frac{(1-S)\omega_0^2}{36}
\end{equation*}
Solving in powers of $S$ and $\omega_0^2$, a set of constraints can be obtained leading to the following 
conditions:
\begin{equation*}
 b_1=0, \qquad b_2=-\frac{1}{36}, \qquad b_3=0, \qquad b_4+d_4+2d_5=0  
\end{equation*}
The same procedure can be applied to the other projections, eventually leading to:
\begin{align*}
&    c_1=d_1, \qquad c_2=-\frac{1}{36}, \qquad 3c_3-d_4=0, \qquad c_4=0 \\
&    c_1=0, \qquad d_2-c_2=\frac{1}{18}, \qquad 3c_3-d_4=0, \qquad c_4=0 \\
&    d_3=\frac{1}{18}, \qquad d_4+d_5=-\frac{1}{36}
\end{align*}
Altogether, the relations between coefficients constitute an overdetermined system of equations. Nevertheless,
it is straightforward to show that equation \eqref{2ndorder} is the only possible solution.

At the fourth order in derivatives, and taking into account parity and time-reversal transformations, our {\em ansatz} 
for the stress-energy tensor is given by:
\begin{align}
    \langle\widehat{T}_{\mu\nu}\rangle^{(4)} =\ & \rho^{(4)}u_\mu u_\nu - p^{(4)}\Delta_{\mu\nu}+ \mathcal{A}^{(2)}A_\mu A_\nu + W^{(2)} \omega_\mu\omega_\nu + 2q^{(2)}u_{(\mu}l_{\nu)} + p_R^{(2)}R_{\mu\nu} +\nonumber \\
    & 2p_{Ru}^{(2)}u_{(\mu}R_{\nu)}^{\ \ \alpha}u_\alpha +2 p_{RA}^{(0)}A_{(\mu}R_{\nu)}^{\ \ \alpha}A_\alpha +2 p_{R\omega}^{(0)}\omega_{(\mu}R_{\nu)}^{\ \ \alpha}\omega_\alpha + p_{RR}^{(0)}R_\mu^{\ \ \alpha} R_{\nu\alpha} +\nonumber \\
    &  p_{RuRu}^{(0)}R_{\mu}^{\ \ \alpha}u_\alpha R_{\nu}^{\ \ \beta}u_{\beta} + 2p_{A\omega}^{(2)}A_{(\mu}\omega_{\nu)} + 2p_{uRl}^{(0)}u_{(\mu}R_{\nu)}^{\ \ \alpha}l_{\alpha} + 2p_{lRu}^{(0)}l_{(\mu}R_{\nu)}^{\ \ \alpha}u_{\alpha}
    \label{GeneralFormOrder4}
\end{align}
where the scalar functions read:
\begin{align*}
    \rho^{(4)}=\ &b_1A^4+b_2RA^2+b_3R\omega^2+b_4R^2+b_5\omega^4+b_6A^2\omega^2+b_7(R^{\mu\nu}u_{\mu}u_{\nu})^2+\nonumber\\
    &b_8(R^{\mu\nu}u_{\mu}u_{\nu})R+b_9(R^{\mu\nu}u_{\mu}u_{\nu})A^2+b_{10}(R^{\mu\nu}u_{\mu}u_{\nu})\omega^2+b_{11}(R^{\mu\nu}u_{\mu}l_{\nu})+ \nonumber\\
    &b_{12}R^{\mu\nu}R_{\mu\nu}+b_{13}(R^{\mu\nu}R_{\mu\alpha}u^{\alpha}u_{\nu})+b_{14}(R^{\mu\nu}A_{\mu}A_{\nu})+b_{15}(R^{\mu\nu}\omega_{\mu}\omega_{\nu})+b_{16}(A^{\mu}\omega_{\nu})^2,
\end{align*}
\begin{align*}
    p^{(4)}=\ &c_1A^4+c_2RA^2+c_3R\omega^2+c_4R^2+c_5\omega^4+c_6A^2\omega^2+c_7(R^{\mu\nu}u_{\mu}u_{\nu})^2+\nonumber\\
    &c_8(R^{\mu\nu}u_{\mu}u_{\nu})R+c_9(R^{\mu\nu}u_{\mu}u_{\nu})A^2+c_{10}(R^{\mu\nu}u_{\mu}u_{\nu})\omega^2+c_{11}(R^{\mu\nu}u_{\mu}l_{\nu})+ \nonumber\\
    &c_{12}R^{\mu\nu}R_{\mu\nu}+c_{13}(R^{\mu\nu}R_{\mu\alpha}u^{\alpha}u_{\nu})+c_{14}(R^{\mu\nu}A_{\mu}A_{\nu})+c_{15}(R^{\mu\nu}\omega_{\mu}\omega_{\nu})+c_{16}(A^{\mu}\omega_{\nu})^2.
\end{align*}
and:
\begin{align*}
    \mathcal{A}^{(2)}&=d_1A^2+d_2\omega^2+d_3R+d_4R^{\mu\nu}u_\mu u_\nu \\
    W^{(2)}&=e_1A^2+e_2\omega^2+e_3R+e_4R^{\mu\nu}u_\mu u_\nu  \\
    q^{(2)}&=f_1A^2+f_2\omega^2+f_3R+f_4R^{\mu\nu}u_\mu u_\nu   \\
    p_R^{(2)}&=g_1A^2+g_2\omega^2+g_3R+g_4R^{\mu\nu}u_\mu u_\nu  \\
    p_{Ru}^{(2)}&=h_1A^2+h_2\omega^2+h_3R+h_4R^{\mu\nu}u_\mu u_\nu  \\
    p_{A\omega}^{(2)}&=i_1(A\cdot\omega)
\end{align*}
while the order zero coefficients are pure numbers. It should be pointed out that this is not the most general 
expression one can write in a generic space-time because contributions from the Weyl tensor and from the gradients 
of the Ricci tensor are not considered because the corresponding coefficients cannot be determined from our 
distilled results. By using the same method described for the second order calculation some constraints 
on the coefficients can be obtained. 
In particular using \eqref{GeneralFormOrder4} and the distilled result in the CEU with non-vanishing vorticity we 
are able to determine all the coefficients appearing in flat space-time, effectively recovering \eqref{minkset}.

We now come to the minimal coupling case. The general form of the stress-energy tensor at second order in derivatives 
without vorticity is given by \eqref{2AdSMinimal}. We use the equality $R_{\mu\nu}=(R/4)g_{\mu\nu}$ 
to rewrite it in AdS as:
\begin{equation}\label{minasdansatz}
    \beta^2\langle\widehat{T}_{\mu\nu}\rangle^{(2)} = (b_1A^2+\Tilde{b}_3 R)u_{\mu}u_{\nu}-(c_1A^2+\Tilde{c}_3R)
    \Delta_{\mu\nu}+d_1A_{\mu}A_{\nu}
\end{equation}
where:
\begin{equation*}
    \Tilde{b}_3 = b_3+\frac{1}{4}(b_4+d_4+2d_5), \qquad \qquad \Tilde{c}_3 = c_3+\frac{1}{4}(c_4-d_4)
\end{equation*}
The above expression is to be matched to the distilled result in AdS, which is given by the \eqref{AmbrusTmunu}
with scalar functions in \eqref{rhods} and \eqref{othersds}. The first matching condition arises from the double 
projection along the four velocity:
\begin{equation*}
    \beta^2u^\mu u^\nu\langle\widehat{T}_{\mu\nu}\rangle_{\text{AdS}}^{(2)} = -
    \frac{\kappa^2(5+3r^2\kappa^2)}{24(1+r^2\kappa^2)} \qquad \qquad \beta^2u^\mu u^\nu\langle\widehat{T}_{\mu\nu}\rangle^{(2)} 
    = b_1A^2+\Tilde{b_2}R
\end{equation*}
leading to:
\begin{equation*}
    b_1=-\frac{1}{12},  \qquad \qquad  \Tilde{b}_3=\frac{5}{288}
\end{equation*}
Projecting the general {\em ansatz} \eqref{minasdansatz} and the distilled stress-energy tensor in AdS
along the two other different independent directions we obtain a set of constraints yielding:
\begin{equation*}
    c_1=\frac{1}{18}, \qquad \tilde{c}_3=-\frac{1}{864}, \qquad d_1=\frac{1}{12}
\end{equation*}
At fourth order, using the same property of the Ricci tensor in AdS, we have the following ansatz:
\begin{equation*}
    \langle\widehat{T}_{\mu\nu}\rangle^{(4)} = (b_1A^4+\Tilde{b}_2RA^2+\Tilde{b}_4R^2)u_{\mu}u_{\nu}-(c_1A^4+\Tilde{c}_2RA^2+\Tilde{c}_4R^2)\Delta_{\mu\nu}+(d_1A^3+\Tilde{d}_3R)A_{\mu}A_{\nu}
\end{equation*}
where:
\begin{align*}
    \tilde{b}_2&= b_2 +\frac{1}{4}(b_9+b_{14}+h_1-g_1)  \\
    \tilde{b}_4&=b_4+\frac{1}{16}(b_7+4b_8+b_{12}+b_{13}+4h_3+h_4+p_{RuRu}^{(0)}-4g_3-g_4-p_{RR}^{(0)}) \\
    \tilde{c}_2&=c_2+\frac{1}{4}(c_9+c_{14}-g_1) \\
    \tilde{c}_4&=c_4+\frac{1}{16}(c_7+4c_8+c_{12}+c_{13}-4g_3-g_4-p_{RR}^{(0)}) \\
    \tilde{d}_3&=d_3+\frac{1}{4}(d_4+p_{RA}^{(0)})
\end{align*}
The obtained relations read:
\begin{align*}
    b_1&=-\frac{11}{480\pi^2} \qquad\qquad \Tilde{b}_2=\frac{37}{5760\pi^2} \qquad\qquad \Tilde{b}_4=-\frac{71}{69120\pi^2} \\
    c_1&=-\frac{19}{1440\pi^2} \qquad\qquad \Tilde{c}_2=\frac{23}{17280\pi^2} \qquad\qquad \Tilde{c}_4=-\frac{139}{207360\pi^2} \\
    d_1&=-\frac{1}{48}, \qquad\qquad \Tilde{d}_3=\frac{1}{576}  
\end{align*}

\end{document}